\documentclass[lettpaper,11pt]{article}
\usepackage{fullpage}
\usepackage[usenames,dvipsnames]{xcolor}

\RequirePackage[OT1]{fontenc}
\RequirePackage{amsthm,amsmath,amssymb,mathrsfs}
\RequirePackage[sort,round]{natbib}
\RequirePackage[breaklinks,colorlinks]{hyperref}
\RequirePackage{hypernat}

\hypersetup{colorlinks=false}

\usepackage{verbatim}
\usepackage{soul}
\usepackage{graphicx}
\usepackage{wrapfig,lipsum}

\DeclareMathOperator{\Cov}{Cov}

\DeclareMathOperator{\tr}{tr}

\newcommand{\field}[1]{\mathbb{#1}}
\newcommand{\R}{\field{R}}

\newcommand{\E}{\field{E}}
\newcommand{\C}{\field{C}}

\newcommand{\h}[1]{\boldsymbol{#1}}
\newcommand{\diag}{\mbox{diag}}

\newcommand{\hzero}{{\boldsymbol{0}}}

\newcommand{\hA}{{\boldsymbol{A}}}
\newcommand{\hf}{{\boldsymbol{f}}}
\newcommand{\hB}{{\boldsymbol{B}}}
\newcommand{\hC}{{\boldsymbol{C}}}
\newcommand{\hD}{{\boldsymbol{D}}}
\newcommand{\hE}{{\boldsymbol{E}}}
\newcommand{\hF}{{\boldsymbol{F}}}
\newcommand{\hH}{{\boldsymbol{H}}}
\newcommand{\hI}{{\boldsymbol{I}}}
\newcommand{\hP}{{\boldsymbol{P}}}
\newcommand{\hQ}{{\boldsymbol{Q}}}
\newcommand{\hR}{{\boldsymbol{R}}}
\newcommand{\hS}{{\boldsymbol{S}}}

\newcommand{\hV}{{\boldsymbol{V}}}
\newcommand{\hW}{{\boldsymbol{W}}}
\newcommand{\hX}{{\boldsymbol{X}}}
\newcommand{\hY}{{\boldsymbol{Y}}}
\newcommand{\hZ}{{\boldsymbol{Z}}}

\newcommand{\hu}{{\boldsymbol{u}}}

\newcommand{\halpha}{{\boldsymbol{\alpha}}}
\newcommand{\hbeta}{{\boldsymbol{\beta}}}

\newcommand{\hgamma}{{\boldsymbol{\gamma}}}

\newcommand{\vect}{\mathrm{vec}}

\theoremstyle{plain}
\newtheorem{theorem}{Theorem}
\newtheorem{corollary}[theorem]{Corollary}
\newtheorem{proposition}[theorem]{Proposition}
\newtheorem{lemma}[theorem]{Lemma}

\theoremstyle{definition}

\theoremstyle{definition}

\makeatletter

\newcommand{\Rmnum}[1]{\expandafter\@slowromancap\romannumeral #1@}
\makeatother

\newcommand{\comments}[1]{}

\newcounter{regular} 
\newcounter{stable} 

\begin{document}

\centerline{{\Large\textbf{Autoregressive Models for Matrix-Valued Time Series}}}

\bigskip \centerline{Rong Chen, Han Xiao and Dan Yang\footnote{
Rong Chen is Professor at Department of
    Statistics, Rutgers University, Piscataway, NJ 08854. E-mail:
    rongchen@stat.rutgers.edu.
Han Xiao is Associate Professor at
    Department of Statistics, Rutgers University, Piscataway, NJ
    08854. E-mail: hxiao@stat.rutgers.edu.
 Dan Yang
    is Assistant Professor at Department of Statistics, Rutgers
    University, Piscataway, NJ 08854. E-mail:
    dyang@stat.rutgers.edu.
Rong Chen is the
    corresponding author. Chen's research is supported
in part by National Science Foundation
grants DMS-1503409, DMS-1737857 and IIS-1741390. Xiao's research is supported in part by a research grant from NEC Labs America.
Yang's research is
suppored in part under NSF grant IIS-1741390.
}} \medskip \centerline{Rutgers University}

\bigskip
\smallskip

\centerline{\bf{Abstract}}

\noindent
In finance, economics and many other fields, observations in a matrix form are often generated over time. For example, a set of key economic indicators are regularly reported in different countries every quarter.  The observations at each quarter neatly form a matrix and are observed over consecutive quarters. Dynamic transport networks with observations generated on the edges can be formed as a matrix observed over time.  Although it is natural to turn the matrix observations into long vectors, then use the standard vector time series models for analysis, it is often the case that the columns and rows of the matrix represent different types of structures that are closely interplayed. In this paper we follow the autoregression for modeling time series and propose a novel matrix autoregressive model in a bilinear form that maintains and utilizes the matrix structure to achieve a substantial dimensional reduction, as well as more interpretability. Probabilistic properties of the models are investigated.  Estimation procedures with their theoretical properties are presented and demonstrated with simulated and real examples.


\newpage
\section{Introduction}

Multivariate time series is a classical area in time series analysis,
and has been extensively studied in the literature \citep[see][for an
overview]{hannan:1970,lutkepohl:2005,tsay:2014,tiao:1981}. Recently
there has been an emerging interest in modeling high dimensional time
series. Roughly speaking these works fall into two major categories:
(i) vector autoregressive modeling with regularization \citep[][among
others]{davis:2012,basu:2015,guo:2015,han:2015,han:2016,nicholson:2015,song:2011b,kock:2015,negahban:2011,nardi:2011},
and (ii) statistical or dynamic factor models \citep[][among
others]{bai:2002,forni:2005,lam:2011,lam:2012,fan:2013,wang:2018}.  In
most of these studies, the multiple observations at each time point
are treated as a vector.

Although it has been conventional to treat multiple observations as a
vector, often the inter-relationship among the time series exhibits
some more structure. For example, \cite{hallin2011dynamic} studied
subpanel structures in multivariate time series, and \cite{tsai:2010}
considered group constraints among the time series. When the time
series are collected under the intersections of two classifications,
they naturally form matrices.

\begin{figure}[h]
  \centering
  \includegraphics[width=12.5cm]{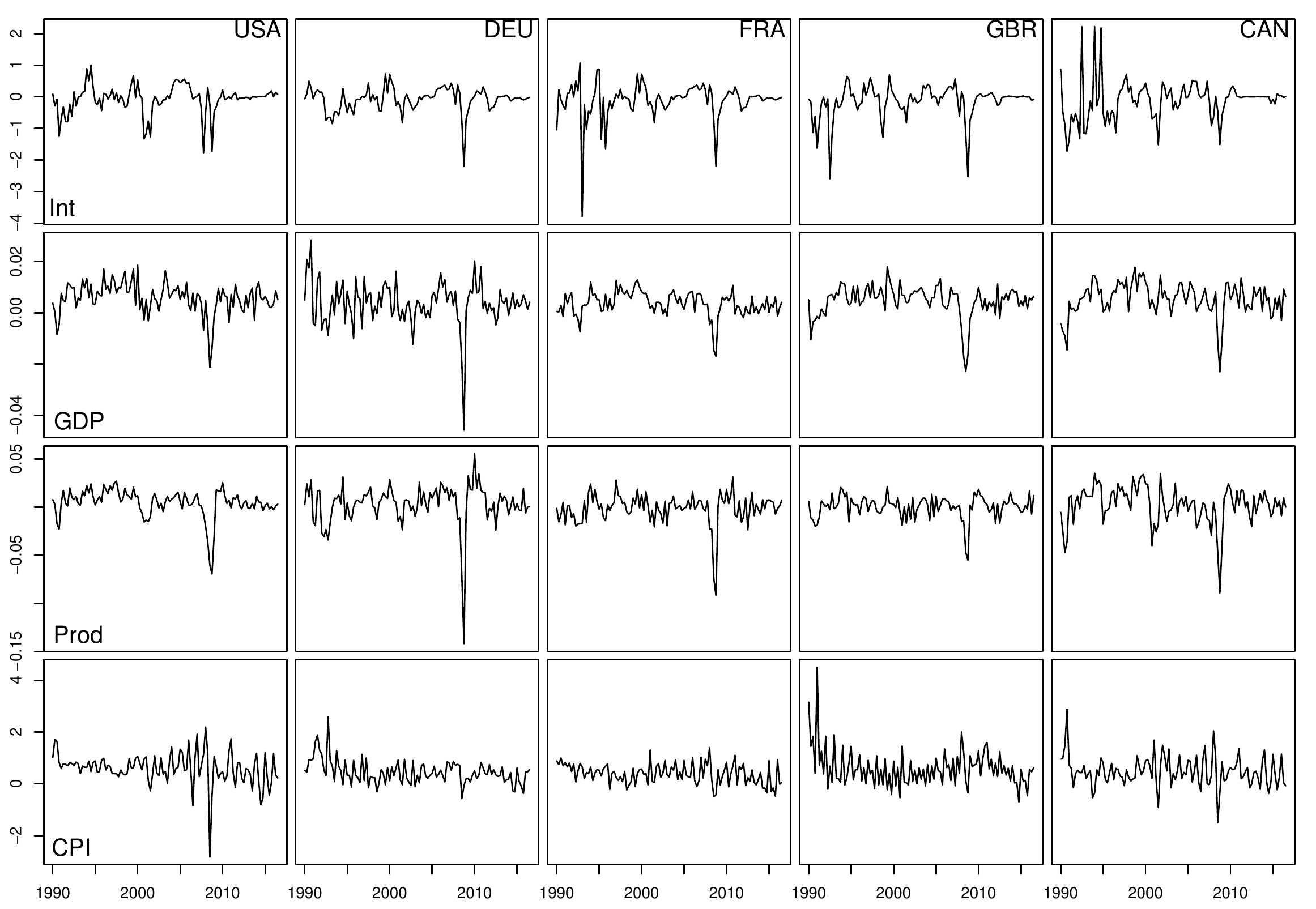}
  \caption{Time series of four economic indicators (first differenced 3 month interbank 
  interest rate, GDP growth (log difference), Total Manufacture 
  Production growth (log difference), and Total CPI growth from last period) from five countries.}
  \label{fig:1}
\end{figure}

In Figure~\ref{fig:1}, we plot four economic indicators from five
countries, resulting in a $4\times 5$ matrix observed at each time
point. In this example, the rows and columns correspond to different
classifications (economic indicators and countries).  Univariate time
series analysis would deal with individual time series separately
(e.g. US interest rate, or UK GDP). Panel time series analysis deals
with one row at a time (e.g. interest rates of the five countries), or
one column at a time (e.g. all economic indicators of US).  Obviously
every time series is related to all other time series in the matrix
and we wish to model them jointly. It is reasonable to assume that the
same economic indicator from different countries (the rows) form a
strong relationship, and at the same time economic indicators from the
same country (the columns) also naturally move together closely. Hence
there is a strong structure in the relationship among the time
series. If the matrices are concatenated into vectors, the underlying
structure is lost with significant impacts on the model {complexity}
and interpretations.



In this paper we propose to model
the matrix-valued time series under the autoregressive framework with
a bilinear form. Specifically, in this model, the conditional mean of
the matrix observation at time $t$ is obtained by multiplying the
previous observed matrix at time $t-1$ from both left and right by two
autoregressive coefficient matrices. Let $\hX_t$ be the $m\times n$ matrix observed at time $t$, our model takes the form
\begin{equation*}
  \h{X}_t = \h{A}\h{X}_{t-1}\h{B}' + \h{E}_t,
\end{equation*}
It can be extended to involve
the previous $p$ observed matrices to form an order $p$ autoregressive
model. If it involves $p$ previously observed matrices, we call it the
{\it matrix autoregressive model of order $p$}, with the acronym
MAR($p$). Compared with the traditional vector autoregressive models,
our approach has two advantages: (i) it keeps the original matrix
structure, and its two coefficient matrices have corresponding
interpretations; and (ii) it reduces the number of parameters in the
model significantly.

Similar bilinear models have been used in regression settings. For
instance, \cite{Wang+2018} considered the regression with
matrix-valued covariates for low-dimensional data. \cite{Zhao+2014a}
studied the bilinear regression with sparse coefficient vectors under
high-dimensional setting. \cite{Zhou+2013a} and \cite{Raskutti+2015}
mainly addressed the multi-linear regression with general tensor
covariates.

A major objective of our model is to take full advantage of the original matrix structure, so that the model is naturally interpretable. A similar concern has emerged in the econometrics literature when studying a large panel of data consisting of blocks. Hierarchical or multi-level factor models have been introduced to capture both the within-block and between-block variations \citep{moench2013dynamic,diebold2008global,giannone2008nowcasting}. Our model shares an interpretation as the hierarchical autoregression, which will be detailed in Section~\ref{sec:interpretation}.

Our model leads to a substantial dimension reduction as compared with a direct vector autoregressive model ($m^2+n^2$ vs $m^2n^2$). However, when the matrix observations themselves have large dimensions, it would be desirable to impose further constraints so that a greater dimension reduction can be achieved. There are a number of possible approaches. First, we may require both $\hA$ and $\hB$ to be sparse, and carryout the estimation with a $\ell_1$ penalty. Second, we can consider the additional assumption that both $\hA$ and $\hB$ are of low ranks. If we fix one of $\hA$ and $\hB$ and consider the other as the parameter, the model can be viewed as a reduced rank regression \citep{anderson:1951,izenman:1975}. This second approach is also related to the recent work \cite{wang:2018}, which studied the factor models for matrix-valued time series. These extensions are beyond the scope of this paper, and we would leave them for further research.

Since the error term in the matrix AR model is also
 a matrix, its (internal) covariances form a 4-dimensional tensor.  Here we
also consider to exploit the matrix structure to reduce the dimensionality of this
covariance tensor, by separating the row-wise and column-wise
dependencies of the error matrix.

In this paper we investigate some probabilistic properties of the
proposed model.  Several estimators for MAR(1) model are developed,
with different computing algorithms, under different assumptions on
the error covariance structure. Their asymptotic properties are
investigated.  We also compare the efficiencies of the estimators.  In
addition, the finite sample performances of the estimators are
demonstrated through simulation studies. The matrix time series of
four economic indicators from five countries, shown in
Figure~\ref{fig:1}, is analyzed in detail.

The rest of the paper is organized as follows. We introduce the
autoregressive model for matrix-valued time series in
Section~\ref{sec:model}, along with some of its
probabilistic properties.  The estimation procedures are presented  in
Section~\ref{sec:est}. Statistical inferences and the asymptotic
properties of the estimators will be considered in
Section~\ref{sec:asy}. Numerical studies are carried out in
Section~\ref{sec:num}. Section~\ref{sec:con} contains a short summary.
All the proofs are collected in Appendix.

\section{Autoregressive Model for Matrix-Valued Time Series}
\label{sec:model}

Consider a time series of length $T$, in which at each time $t$, a
$m\times n$ matrix $\h{X}_t$ is observed. Here we use $\h{X}$ in
boldface to emphasize the fact that it is a matrix. Let
$\mathrm{vec}(\cdot)$ be the vectorization of a matrix by
stacking its columns. The traditional vector autoregressive model
(VAR) of order 1 is directly applicable for $\vect(\hX_t)$. That is,
\begin{equation}
  \label{eq:var1}
  \mathrm{vec}(\h{X}_t) = \Phi\vect(\hX_{t-1}) + \vect(\hE_t).
\end{equation}

It is immediately seen that the roles of rows and columns are mixed in
the VAR model in \eqref{eq:var1}. Using the example shown in
Figure~\ref{fig:1}, the VAR model in \eqref{eq:var1} fails to
recognize the strong connections within the columns (same country) and
within the rows (same indicator).  The (large) $mn\times mn$
coefficient matrix $\Phi$ does not have any assumed structure; and the
model does not fully utilize the matrix structure, or any prior
knowledge of the potential relationship among the time series.  The
coefficient matrix $\Phi$ is also very difficult to interpret.

To overcome the drawback of the direct VAR modeling that requires
vectorization, and {to} take advantage of the original matrix
structure, we propose the {\it matrix autoregressive model (of order
  1)}, denoted by MAR(1), in the form
\begin{equation}
  \label{eq:mar1}
  \h{X}_t = \h{A}\h{X}_{t-1}\h{B}' + \h{E}_t,
\end{equation}
where $\h{A}=(a_{ij})$ and $\h{B}=(b_{ij})$ are $m\times m$ and
$n\times n$ autoregressive coefficient matrices, and
$\h{E}_t=(e_{t,ij})$ is a $m\times n$ matrix white noise. Clearly
the model can be extended to an order $p$ model in the form
\[
  \h{X}_t = \h{A}_1\h{X}_{t-1}\h{B}_1' + \cdots +
\h{A}_p\h{X}_{t-p}\h{B}_p' + \h{E}_t.
\]
We will defer the interpretations of $\hA$ and $\hB$ in (\ref{eq:mar1})
to Section 2.1.

The MAR(1) model
in (\ref{eq:mar1}) can be represented in the form of a vector
autoregressive model
\begin{equation}
  \label{eq:ovar1}
  \mathrm{vec}(\h{X}_t) = (\hB\otimes\hA)\vect(\hX_{t-1}) + \vect(\hE_t),
\end{equation}
where $\otimes$ denotes the matrix Kronecker product.  A thorough
discussion of the Kronecker product and its relationship with linear
matrix equations can be found in Chapter~4 of \cite{horn:1994}. In the Appendix, we collect
some basic properties of the Kronecker product in Proposition~\ref{thm:kronecker}. The representation
\eqref{eq:ovar1} means that the MAR(1) model can be viewed as a
special case of the classical VAR(1) model in (\ref{eq:var1}), with
its autoregressive coefficient matrix given by a Kronecker product. On
the other hand, comparing \eqref{eq:var1} and \eqref{eq:ovar1}, we see
that the MAR(1) requires $m^2+n^2$ coefficients as the entries of
$\hA$ and $\hB$, while an unrestricted VAR(1) needs $m^2n^2$ coefficients
for $\Phi$. Apparently the latter can be much larger when both $m$ and
$n$ are large.

There is an obvious identifiability issue with the MAR(1) model
in \eqref{eq:mar1}, regarding the two coefficient matrices $\hA$ and
$\hB$. The model remains unchanged if the two matrices $\hA$ and $\hB$
are divided and multiplied by the same nonzero constant
respectively. To avoid ambiguity,
we use the convention that $\hA$ is normalized so that its
Frobenius
norm is one.
On the other hand, the uniqueness always holds for the
Kronecker product $\hB\otimes\hA$.

The error matrix sequence $\{\hE_t\}$ is assumed to be a matrix white
noise, i.e. there is no correlation between $\hE_t$ and $\hE_s$ as
long as $t\neq s$. But $\hE_t$ is still allowed to have concurrent
correlations among its own entries.  As a matrix, its covariances form
a 4-dimensional tensor, which is difficult to express. In the
following we will discuss it in the form of
$\Sigma=\Cov(\vect(\hE_t))$, a $(mn)\times (mn)$ matrix. As the simplest case, we may assume the entries of $\hE_t$ are independent so that
$\Cov(\vect(\hE_t))$ is a diagonal matrix; and in general, we allow them to have arbitrary correlations. We also consider
a structured covariance matrix
\begin{equation}
  \label{eq:ocov}
  \Cov(\vect(\hE_t)) = \Sigma_c\otimes\Sigma_r,
\end{equation}
where $\Sigma_r$ and $\Sigma_c$ are $m\times m$ and $n\times n$
symmetric positive definite matrices. Under normality, this is
equivalent to assuming $\hE_t = \Sigma_r^{1/2}\h{Z}_t\Sigma_c^{1/2}$,
where all the entries of $\h{Z}_t$ are independent, and following the
standard Normal distribution.  Therefore, $\Sigma_r$ corresponds to
row-wise covariances and $\Sigma_c$ introduces column-wise
covariances.

\noindent {\bf Remarks: \ } There are many possible extensions of the
model. For example, the model can be extended to have multiple lag-one
autoregressive terms. That is
\begin{equation}
\label{multi-term}
  \h{X}_t = \h{A}_1\h{X}_{t-1}\h{B}_1' + \cdots +
\h{A}_d\h{X}_{t-1}\h{B}_d' + \h{E}_t.
\end{equation}
This is still an order-1 autoregressive model, but with more parallel terms.
In the stacked vector form, it corresponds to
\[
  \mathrm{vec}(\h{X}_t) = \left(\sum_{i=1}^d \hB_i\otimes\hA_i\right)\vect(\hX_{t-1})
+ \vect(\hE_t).
\]
Such a structure provides more flexibility to capture the different
interactions among rows and columns of the matrix time series, 
though it becomes more challenging, since there is obviously a more severe identifiability issue.

In this paper we focus on MAR(1) model \eqref{eq:mar1} in all our discussions. Extensions will be investigated elsewhere.

\subsection{Model interpretations}
\label{sec:interpretation}

The MAR(1) model is not a straightforward model. A thorough
discussion of the interpretations of its coefficient
matrices is needed. Here we offer interpretations from three different angles.

First, in model (\ref{eq:mar1}), the left matrix $\h{A}$ reflects
row-wise interactions, and the right matrix $\h{B}'$ introduces
column-wise dependence, and therefore the conditional mean in
\eqref{eq:mar1} combines the row-wise and column-wise interactions. It
is easier to see how the coefficient matrices $\hA$ and $\hB$ reflects
the row and column structures by looking at a few special cases.

To isolate the effect from the bilinear form in
\eqref{eq:mar1}, let us assume $\hA=\hI$. Then the model reduces to
$$\hX_t = \hX_{t-1}\hB' + \hE_t.$$
Consider the example shown in Figure~\ref{fig:1}, using columns
for countries and rows for economic indicators.
The conditional expectation of the first column of $\hX_t$ is given by
\begin{align*}
  &\quad\hbox{USA} \qquad\qquad\quad \hbox{USA} \quad\qquad\qquad\;\;\,
\hbox{DEU}\qquad\qquad\qquad\qquad\,  \hbox{CAN}\\
  &\begin{pmatrix}
    \hbox{Int} \\
    \hbox{GDP} \\
    \hbox{Prod} \\
    \hbox{CPI}
  \end{pmatrix}_t =
  b_{11}\begin{pmatrix}
    \hbox{Int} \\
    \hbox{GDP} \\
    \hbox{Prod} \\
    \hbox{CPI}
  \end{pmatrix}_{t-1} +
  b_{12}\begin{pmatrix}
    \hbox{Int} \\
    \hbox{GDP} \\
    \hbox{Prod} \\
    \hbox{CPI}
  \end{pmatrix}_{t-1}+\cdots+
  b_{1n}\begin{pmatrix}
    \hbox{Int} \\
    \hbox{GDP} \\
    \hbox{Prod} \\
    \hbox{CPI}
  \end{pmatrix}_{t-1},
\end{align*}
which means that at time $t$, the conditional expectation of an
economic indicator of one country is a linear combinations of the same
indicator from all countries at $t-1$, and this linear combination is
the same for different indicators. Therefore, this model (and
$\hB$) captures the column-wise interactions, i.e. interactions among the
countries. However, the interactions are refrained within each
indicator. There are no interactions among the indicators.

On the other hand, if we let $\hB=\hI$ in model \eqref{eq:mar1}, then
a similar interpretation can be obtained, where the matrix $\hA$
reflects the row-wise interactions, i.e. interactions among the
economic indicators within each country. There are no interactions
among the countries.

Second, we can interpret the model from a row-wise and column-wise VAR
model point of view. For example, if $\hB=\hI$, then the model becomes
$$\hX_t = \hA\hX_{t-1} + \hE_t.$$ In this case, each column of $\hX_t$
follows
\[
\hX_{t,\cdot j}=\hA\hX_{t-1,\cdot j}+\hE_{t, {\cdot j}}, \mbox{ \ \ }
j=1,\ldots n.
\]
That is, each {column} of $\hX_t$ follows the same VAR(1) model of
dimension {$m$}. Specifically, for the first two columns in the example
of Figure~\ref{fig:1}, the models are

\begin{align*}
  &\hbox{~~~~USA} \qquad\qquad\;\; \hbox{USA} \qquad\qquad\;\; \hbox{USA}\;\;
&
&\hbox{~~~~DEU} \qquad\qquad\;\; \hbox{DEU} \qquad\qquad\;\; \hbox{DEU}\;\; \\
  &\begin{pmatrix}
    \hbox{Int} \\
    \hbox{GDP} \\
    \hbox{Prod} \\
    \hbox{CPI}
  \end{pmatrix}_t =
  \hA\begin{pmatrix}
    \hbox{Int} \\
    \hbox{GDP} \\
    \hbox{Prod} \\
    \hbox{CPI}
  \end{pmatrix}_{t-1} +
    \begin{pmatrix}
    \hbox{e\_Int} \\
    \hbox{e\_GDP} \\
    \hbox{e\_Prod} \\
    \hbox{e\_CPI}
  \end{pmatrix}_{t}
& \mbox{and\ \ }  &\begin{pmatrix}
    \hbox{Int} \\
    \hbox{GDP} \\
    \hbox{Prod} \\
    \hbox{CPI}
  \end{pmatrix}_t =
  \hA\begin{pmatrix}
    \hbox{Int} \\
    \hbox{GDP} \\
    \hbox{Prod} \\
    \hbox{CPI}
  \end{pmatrix}_{t-1} +
    \begin{pmatrix}
    \hbox{e\_Int} \\
    \hbox{e\_GDP} \\
    \hbox{e\_Prod} \\
    \hbox{e\_CPI}
  \end{pmatrix}_{t}
\end{align*}

In other words, for each country, its economic indicators follow a VAR(1)
model (of its own past) of dimension $m$;
and different countries would follow the same VAR(1) model.

If $\hA=\hI$, then $$\hX_t = \hX_{t-1}\hB' + \hE_t.$$ In this case,
each row of $\hX_t$ (same indicator from different countries) would
follow a VAR(1) model. And the coefficient matrices corresponding to
different rows would be the same.

Obviously both models
$\hX_t = \hA\hX_{t-1} + \hE_t$ and  $\hX_t = \hX_{t-1}\hB' + \hE_t$
are too restrictive. It is difficult to reason that Germany's
economic indicators follow the same model as the US's, and there is
no interaction between Germany and US. There are two possible ways to
add flexibility. One can assume an additive interaction structure to make
the model as
\[
\hX_t = \hA\hX_{t-1} + \hX_{t-1}\hB' + \hE_t,
\]
which is essentially a special case of the multi-term model in
(\ref{multi-term}) with $d=2$ and $\hB_1=\hI$ and $\hA_2=\hI$.  Or we
can assume a one-term multiplicative interaction structure, which
leads to MAR(1). Of course, one can also use
\[
\hX_t = \hA_1\hX_{t-1} + \hX_{t-1}\hB_1'+\hA_2\hX_{t-1}\hB_2' + \hE_t,
\]
similar to the model with main effects plus two-way interactions. In
this paper we choose to work on MAR(1) in (\ref{eq:mar1}).

The third way to interpret MAR(1) is through a defined hierarchical structure. Multi-level or hierarchical factor models have been introduced in the econometric literature to study a large panel of data consisting of blocks or even sub-blocks \citep{moench2013dynamic,diebold2008global,giannone2008nowcasting}. Here we illustrate that our model shares a similar interpretation as hierarchical autoregression.
Let $\hY_{t-1}=\hX_{t-1}\hB'$. It would be the prediction of $\hX_t$
(or the conditional mean) if $\hA=\hI$. Since each column of
$\hY_{t-1}$ is based on the linear combination of all columns of
$\hX_{t-1}$ with no row (indicator) interaction, we can view each entry
in $\hY_{t-1}$ as the {\bf globally adjusted indicator}. For
example, $\hY_{t-1, GDP, US}$ is a linear combination of the GDPs of
all countries at time $t-1$. Next, we consider
$\hZ_{t-1}=\hA\hY_{t-1}$. This would be the prediction of $\hX_t$ {if}
the model is $\hX_t=\hA\hY_{t-1}+\hE_t$. It replaces each entry
(indicator) in $\hX_{t-1}$ by its corresponding globally adjusted indicator in
$\hY_{t-1}$. Each entry in $\hZ_{t-1}$ is a linear combination of the
adjusted indicators from the {\bf same} country. For example,
$\hZ_{t-1,GDP,US}$ is a linear combination of $\hY_{t-1,GDP, US}$,
$\hY_{t-1,INT,US}$ etc. It can be viewed as a second adjustment by
other indicators (within the same country). Putting everything together,
we have
\[
\hX_t=\hZ_{t-1}+\hE_t={\hA\hY_{t-1}}+\hE_t=\hA\hX_{t-1}\hB'+\hE_t.
\]

Note that, if $\hX_t$ follows the MAR(1) in (\ref{eq:mar1}), each entry
$x_{t,ij}$ in $\hX_t$ follows
\[
{x_{t,ij}=\sum_{k_1=1}^m\sum_{k_2=1}^n a_{ik_1}x_{t-1,k_1k_2}b_{jk_2}
+e_{t,ij}}.
\]
Hence $x_{t,ij}$ is controlled only by $i$-th row of $\hA$ and $j$-th
column of $\hB'$. In the example of Figure~\ref{fig:1},
the $i$-th row of $\hA$ can be viewed as the coefficient corresponding
to $i$-th indicator and the $j$-th column of $\hB'$ as the
coefficient corresponding to the $j$-th country. Their values can be
interpreted, as we will demonstrate in the real example in Section 4.

The error covariance matrix
$\Cov(\vect(\hE_t)) = \Sigma_c\otimes\Sigma_r$ in (\ref{eq:ocov}),
which consists of all pairwise covariances
$\Cov(e_{t,ij}, e_{t,kl})=\sigma_{c,jl}\sigma_{r,ik}$, has a similar
interpretation. For example, if $\Sigma_r=\hI$, then
$\hE_t = \hZ\Sigma_c^{1/2}$, which implies that the $\Sigma_c$ matrix
captures the concurrent dependence of the columns of shocks in
$\hE_t$.  Note that each row of $\hE_t$ in this case is
$\hE_{t,i \cdot}=\hZ_{t,i \cdot}\Sigma_c^{1/2}$. Hence
$\Cov(\hE_{t,i \cdot})=\Sigma_c$ for all rows $i=1,\dots, m$, and
therefore $\Sigma_c$ captures the covariance among the (column)
elements in each row. In parallel, $\Sigma_r$ captures the concurrent
dependence among the rows of shocks in $\hE_t$.

\subsection{Probabilistic properties of MAR(1)}

For any square matrix $\h{C}$, we use $\rho(\h{C})$ to denote its
spectral radius, which is defined as the maximum modulus of the
(complex) eigenvalues of $\h{C}$. Since the MAR(1) model can be
represented in the form \eqref{eq:ovar1}, we see that \eqref{eq:mar1}
admits a causal and stationary solution if the spectral radius of
$\hB\otimes\hA$, which is the product of the spectral radii of $\hA$
and $\hB$, is strictly less than 1. Hence we have

\begin{proposition}
  \label{thm:causal}
  If $\rho(\hA)\cdot\rho(\hB)<1$, then the MAR(1) model
  \eqref{eq:mar1} is stationary and causal.
\end{proposition}

The detailed proof of the proposition is given in the Appendix.

We remark that the property of being ``stationary and causal'' is
referred to as being ``stable'' in \cite{lutkepohl:2005}. Here we follow the terminology and definitions used in \cite{brockwell:1991}. The condition $\rho(\hA)\cdot\rho(\hB)<1$ will be referred to as the {\it causality condition} in the sequel.

When the condition of Proposition~\ref{thm:causal} is fulfilled, the
MAR(1) model in \eqref{eq:mar1} has the following causal representation
after vectorization:
\begin{equation}
  \label{eq:mar1_causal}
  \vect(\hX_t)=\sum_{k=0}^\infty {\left(\hB^k\otimes\hA^k\right)}\vect(\hE_{t-k}).
\end{equation}
It follows that the autocovariance matrices of \eqref{eq:mar1} is given
by
\begin{equation}
  \label{eq:autocov}
  \Gamma_k:=\Cov(\vect(\hX_t),\vect(\hX_{t-k})) = \sum_{l=0}^\infty {\left(\hB^{k+l}\otimes\hA^{k+l}\right)
  \Sigma\left(\hB^l\otimes\hA^l\right)'}, \quad k\geq 0,
\end{equation}
{where $\Sigma$ is the covariance matrix of $\vect(\hE_t)$}. The condition $\rho(\hA)\cdot\rho(\hB)<1$ guarantees that the
infinite matrix series is absolutely summable.

\medskip
It is known that a causal and a non-causal VAR(1) can lead to equivalent models under Gaussianity \citep{hannan:1970}. However, if the parameter space of \eqref{eq:var1} is restricted to $\Theta:=\{(\Phi,\Sigma):\;\rho(\Phi)<1, \hbox{ and }\Sigma \hbox{ is nonsingular}\}$, then the VAR(1) model \eqref{eq:var1} is identifiable in the sense that if two sets of parameters $(\Phi_1,\Sigma_1)\in\Theta$ and $(\Phi_2,\Sigma_2)\in\Theta$ generate the same autocovariance functions \eqref{eq:autocov}, then they must be identical. To see this, first note that under causality, the best linear prediction of $\vect(\hX_t)$ by $\vect(\hX_{t-1})$ is $\Phi_i\vect(\hX_{t-1})$, so the prediction error covariance matrices $\Sigma_1$ and $\Sigma_2$ must be the same. Second, under causality, $\Gamma_1=\Phi_i\Gamma_0$. Since $\Sigma_1=\Sigma_2$ are nonsingular, so is $\Gamma_0$. Therefore, both $\Phi_1$ and $\Phi_2$ would equal to $\Gamma_1\Gamma_0^{-1}$. Now we consider the identifiability regarding $\hA$ and $\hB$ over the parameter space $\{(\hA,\hB):\;\|\hA\|_F=1,\;\|\hB\|_F>0\}$. If $\hB_1\otimes\hA_1=\hB_2\otimes\hA_2$, then $\h{M}:=\vect(\hA_1)\vect(\hB_1)'=\vect(\hA_2)\vect(\hB_2)'$. Since $\h{M}$ is a nonzero matrix, and $\|\hA_1\|_F=\|\hA_2\|_F=1$, the uniqueness of the singular value decomposition of $\h{M}$ guarantees that $(\hA_1,\hB_1)=\pm(\hA_2,\hB_2)$, giving the desired identifiability up to a sign change. When $\Sigma$ is defined with the Kronecker product form \eqref{eq:ocov}, the identifiability regarding parameters $\Sigma_r$ and $\Sigma_c$ can be similarly showed if we require both of them to be nonsingular, and $\|\Sigma_r\|_F=1$.

\medskip
We discuss the impulse response function  with orthogonal innovations (oIRF) of the MAR(1) model.
Since MAR(1) is a special case of VAR(1), we follow the definition given in Section~2.10 of \cite{tsay:2014}. 
However, the standard approach requires fixing the order under which 
the innovations are orthogonalized, which is specially difficult to determine for matrix innovations. In this paper we 
adopt a simpler strategy. To obtain the oIRF with a shock at $(i,j)$-th series, we will put that series as the first 
series in the vectorized VAR model, and all other innovations will be orthogonalized with the $(i,j)$-th innovation fixed. The oIRF 
obtained this way actually does not depend on the order of the rest variables and its formulation is simple without the need to
perform Cholesky decomposition. 
We will call the oIRF obtained this way 
{\it the shock-first impulse response function with orthogonal innovations} (s1-oIRF). Specifically, 
let $\Sigma[\,,i]$ be the
$i$-th column of $\Sigma$, and $\sigma_{ij}$ be the $(i,j)$-th entry
of $\Sigma$. Then the s1-oIRF of a unit standard
deviation change in $e_{t,ij}$  (which is $\sigma^{1/2}_{m(j-1)+i,m(j-1)+i}$)  
is given by, in the vectorized form of the original matrix, 
\begin{equation} \label{first-oIRF}
    \hF_{i,j}(k)=\left(\hB^k\otimes\hA^k\right)\Sigma[\,, m(j-1)+i].
\end{equation}
Note that the s1-oIRF in \eqref{first-oIRF} depends on the series to which the shock occurs, just as 
the standard oIRF depends on the order of variables. The accumulated s1-oIRF is in the form
\[
\tilde{\hF}_{i,j}(K)=\left(\sum_{k=0}^K\hB^k\otimes\hA^k\right)\Sigma[\,, m(j-1)+i].
\]

This type of impulse response
function exhibits a special structure if $\Sigma$ has the Kronecker
product form \eqref{eq:ocov}. Let $\Sigma_{r,\cdot i}$ and
$\Sigma_{c, \cdot i}$ be the $i$-th columns of $\Sigma_r$ and
$\Sigma_c$, respectively. Then the effect of a unit standard deviation
change in $e_{t,ij}$ on the future $\vect(\hX_{t+k})$ is given
by $(\hB^k\Sigma_{c,\cdot j})\otimes(\hA^k\Sigma_{r,\cdot i})$.

To see the impact of this formulation, consider the case that a one-standard deviation shock occurs at
$(1,1)$ series. Let $f^{c}_{i}(k)=(\hB^k\Sigma_{c,\cdot 1})[i]$, the
$i$-th element of $\hB^k\Sigma_{c,\cdot 1}$ and
$f^{r}_{j}(k)=(\hA^k\Sigma_{r,\cdot 1})[j]$, the $j$-th element of
$\hA^k\Sigma_{r,\cdot 1}$.  Then the impulse response function of
$(i,j)$-th series at lag $k$ is
\[
f_{i,j}(k)= f^{r}_{i}(k)f^{c}_{j}(k)
=(\hA^k\Sigma_{r,\cdot 1})[i]\cdot(\hB^k\Sigma_{c,\cdot 1})[j].
\]
We further let $\hf^r(k)=[f^r_1(k),\ldots,f^r_m(k)]'$, and
$\hf^c(k)=[f^c_1(k),\ldots,f^c_n(k)]'$. (We use the boldfaced 
$\hf(\cdot)$ here 
to emphesize it is a vector of functions.) We can view
$\hf^{r}(k)$ as the column response function and
$\hf^{c}(k)$ as the row response function.  Using the example shown
in Figure~\ref{fig:1}, if there is a unit standard deviation shock in
the interest rate of US (location $(1,1)$ in the matrix), then its lag
$k$ effect on the four economic indicators for $j$-th country is
\[
[f_{1,j}(k), f_{2,j}(k), f_{3,j}(k), f_{4,j}(k)]'=
[f^{r}_{1}(k), f^{r}_{2}(k),f^{r}_{3}(k),f^{r}_{4}(k)]'f^{c}_{j}(k)=
f^c_j(k)\cdot \hf^r(k),
\]
which has the following interpretations. The effect of the shock on
four economic indicators in the same $j$-th country, a 4-dimensional
vector $[f_{1,j}(k), f_{2,j}(k), f_{3,j}(k), f_{4,j}(k)]'$, is
proportional to the vector
$\hf^{r}(k)$, for
all countries $1\leq j\leq 5$. Hence the five (4-dimensional economic
indicator) vectors corresponding to the five countries are parellel to
each other and only differ by the multiplier $f^{c}_{j}(k)$. This form
of impulse response function implies that the economies of different
countries have a co-movement as responses to a shock, but the impacts
on different countries are of different scales. Similarly, we have
\[
[f_{i,1}(k), \ldots, f_{i,5}(k)]=
f^{r}_{i}(k)\cdot[f^{c}_{1}(k), \ldots ,f^{c}_{5}(k)] = f^{r}_{i}(k)\cdot 
[\hf^c(k)]',\quad 1\leq i\leq 4.
\]
That is, the effect on five countries regarding each economic
indicator, which is a 5-dimensional row vector, is proportional to
$\hf^{c}(k)$, and the four vectors
corresponding to four indicators only differ by lengths
$f^{r}_{i}(k)$.

In general, for a shock that occurs at location $(i,j)$, the lag-$k$ row
and column response functions are
$\hf^{c,j}(k)=\hB^k\Sigma_{c,\cdot j}$ and
$\hf^{r,i}(k)=\hA^k\Sigma_{r,\cdot i}$, respectively.  In fact,
the response function in matrix form in this case is given by the
rank-one matrix:
\[
\hbox{\bf F}_{i,j}(k)=\hf^{r,i}(k)[\hf^{c,j}(k)]'.
\]

\section{Estimation}
\label{sec:est}

\subsection{Projection method}

To estimate the coefficient matrices $\hA$ and $\hB$, our first
approach is to view the MAR(1) model in \eqref{eq:mar1} as the
structured VAR(1) model in \eqref{eq:ovar1}. We first obtain the
maximum likelihood estimate or the least square estimate $\hat\Phi$ of
$\Phi$ in \eqref{eq:var1} without the structure constraint, then we find the
estimators by projecting $\hat\Phi$ onto the space of Kronecker
products under the Frobenius norm:
\begin{equation}
  \label{eq:nkp}
  (\hat\hA_1,\hat\hB_1)=\arg\min_{\hA,\hB}\|\hat\Phi-\hB\otimes\hA\|_F^2.
\end{equation}
This minimization problem is called the {\it nearest Kronecker
  product} (NKP) problem in matrix computation
\citep{vanloan:1993,vanloan:2000}. It turns out that an explicit
solution exists, which can be obtained through a singular value
decomposition (SVD) of a rearranged version of $\hat\Phi$.

Note that the set of all entries in $\hB\otimes \hA$ is exactly the
same as the set of all entries in $\vect(\hA)\vect(\hB)'$. The two matrices
have the same set of elements, and only
differ by the placement of the elements  in the matrices.
Define a re-arrangement
operator $\mathcal{G}: \R^{mn} \times \R^{{mn}} \rightarrow \R^{{m^2}}\times
\R^{{n^2}}$ such that
\[
{\cal G}(\hB\otimes \hA)=\vect(\hA)\vect(\hB)'.
\]
It is easy to see that the operator is a linear operator such that
${\cal G}(\hC_1+\hC_2)={\cal G}(\hC_1)+{\cal G}(\hC_2)$. We also note that
the Frobenius norm of a matrix only depends on the elements in the matrix, but
not the arrangement, hence $||{\cal G}(\hC)||_F=||\hC||_F$. Then we have
\begin{eqnarray*}
\min_{\hA,\hB}\|\hat\Phi-\hB\otimes\hA\|_F^2 & = &
\min_{\hA,\hB}\|{\cal G}(\hat\Phi)-{\cal G}(\hB\otimes\hA)\|_F ^2\\
&=& \min_{\hA,\hB}\|{\cal G}(\hat\Phi) - \vect(\hA)\vect(\hB)'\|_F^2 \\
&=& \min_{\hA,\hB}\|{\tilde \Phi} - \vect(\hA)\vect(\hB)'\|_F^2,
\end{eqnarray*}
where ${\tilde \Phi}={\cal G}(\hat\Phi)$ is the re-arranged
$\hat\Phi$.
It follows that the solution of \eqref{eq:nkp} can be obtained through
\begin{equation*}
  \vect(\hat\hA)\vect(\hat\hB)'=d_1\h{u}_1\h{v}_1',
\end{equation*}
where $d_1$ is the largest singular value of $\tilde\Phi$, and
$\h{u}_1$ and $\h{v}_1$ are the corresponding first left and right
singular vectors, respectively.  By converting the vectors into
matrices, we obtain corresponding estimators of $\hA$ and $\hB$,
denoted by $\hat\hA_1$ and $\hat\hB_1$, with
the normalization that $||\hat{\hA}_1||_F=1$. We call them {\it projection
  estimators}, and will use the acronym PROJ for later references.

We illustrate the re-arrangement operation
with a special case of $m=n=2$.
We first rearrange the entries of the Kronecker product
$\hB\otimes\hA$:
\begin{align*}
    \left[
    \begin{array}{cc|cc}
      b_{11}a_{11}&b_{11}a_{12}&b_{12}a_{11}&b_{12}a_{12} \\
      b_{11}a_{21}&b_{11}a_{22}&b_{12}a_{21}&b_{12}a_{22} \\\hline
      b_{21}a_{11}&b_{21}a_{12}&b_{22}a_{11}&b_{22}a_{12} \\
      b_{21}a_{21}&b_{21}a_{22}&b_{22}a_{21}&b_{22}a_{22} \\
    \end{array}\right]
    \longrightarrow
    \left[
    \begin{array}{cccc}
      b_{11}a_{11}&b_{21}a_{11}&b_{12}a_{11}&b_{22}a_{11} \\
      b_{11}a_{21}&b_{21}a_{21}&b_{12}a_{21}&b_{22}a_{21} \\
      b_{11}a_{12}&b_{21}a_{12}&b_{12}a_{12}&b_{22}a_{12} \\
      b_{11}a_{22}&b_{21}a_{22}&b_{12}a_{22}&b_{22}a_{22}
    \end{array}\right].
\end{align*}
We then rearrange the entries of $\hat\Phi$ in exactly the same way:
\begin{align*}
    \hat\Phi=\left[
    \begin{array}{cc|cc}
      \phi_{11}&\phi_{12}&\phi_{13}&\phi_{14}\\
      \phi_{21}&\phi_{22}&\phi_{23}&\phi_{24}\\\hline
      \phi_{31}&\phi_{32}&\phi_{33}&\phi_{34}\\
      \phi_{41}&\phi_{42}&\phi_{43}&\phi_{44}\\
    \end{array}\right]
    \longrightarrow
    \left[
    \begin{array}{cccc}
      \phi_{11}&\phi_{31}&\phi_{13}&\phi_{33}\\
      \phi_{21}&\phi_{41}&\phi_{23}&\phi_{43}\\
      \phi_{12}&\phi_{32}&\phi_{14}&\phi_{34}\\
      \phi_{22}&\phi_{42}&\phi_{24}&\phi_{44}\\
    \end{array}\right]=:\tilde\Phi.
\end{align*}
By abuse of notation, we omit the hat on each individual
$\hat\phi_{ij}$. Now it is clear that the NKP problem \eqref{eq:nkp}
is equivalent to
$\min_{\hA,\hB}\|\tilde \Phi - \vect(\hA)\vect(\hB)'\|_F^2$.

In fact, by obtaining the first $k$ largest singular values of  ${\tilde \Phi}={\cal G}(\hat\Phi)$ 
and their corresponding $k$-th left and right singular vectors 
$\h{u}_k$ and $\h{v}_k$, respectively, and then converting the vectors into
matrices, we obtain estimators of $\hA_i$ and $\hB_i$ in the multi-term model (\ref{multi-term}),
under proper model assumptions.

Note that this procedure requires the estimation of the $mn\times mn$
coefficient matrix $\Phi$ first. This task is often formidable and
inaccurate with
moderately large $m$ and $n$ and a finite sample size. Hence the
resulting projection estimator may not be very accurate. However, it can
serve as the initial value for a more elaborate iterative
procedure.

\subsection{Iterated least squares}


If we assume the entries of $\h{E}_t$ are i.i.d. normal with mean zero
and a constant variance, the maximum likelihood estimator, denoted by
$\hat\hA_2$ and $\hat\hB_2$, is the solution of the least squares
problem
\begin{equation}
  \label{eq:lse}
  \min_{\hA,\hB}\sum_{t}\|\hX_t-\hA\hX_{t-1}\hB'\|_F^2.
\end{equation}
We refer to this estimator as LSE for the rest of this paper. If the
error covariance matrix is arbitrary, the LSE is still an intuitive
and reasonable estimator. To see the connection between the two
estimators PROJ and LSE, define
\begin{equation}
\label{eq:design.mat}
\begin{aligned}
  \h{\mathcal{Y}}&=[\vect(\hX_2),\vect(\hX_3),\ldots,\vect(\hX_T)], \\
  \h{\mathcal{X}}&=[\vect(\hX_1),\vect(\hX_2),\ldots,\vect(\hX_{T-1})].
\end{aligned}
\end{equation}
The minimization problem \eqref{eq:lse} can be rewritten as
\begin{equation}
  \label{eq:inkp}
  \min_{\hA,\hB}\|\h{\mathcal{Y}}-(\hB\otimes\hA) \h{\mathcal{X}}\|_F^2.
\end{equation}
Comparing \eqref{eq:inkp} and \eqref{eq:nkp}, we see the
problem \eqref{eq:inkp} can be viewed as an inverse NKP
problem. Unfortunately it does not have an explicit SVD solution
\citep{vanloan:2000}.

There is another way to understand the minimization problem
\eqref{eq:lse}. Define
\begin{equation}
  \label{eq:lse_bla}
\begin{aligned}
  \h{\mathfrak{Y}}'&=[\hX_2',\hX_3',\ldots,\hX_T'], \\
  \h{\mathfrak{X}}'_{\hA}&=[\hX_1'\hA',\hX_2'\hA',\ldots,\hX_{T-1}'\hA'].
\end{aligned}
\end{equation}
With these notations, the least squares problem \eqref{eq:lse} is
equivalent to
\begin{equation}
  \label{eq:lse1}
  \min_{\hA}\left\{\min_{\hB}\|\h{\mathfrak{Y}}-\h{\mathfrak{X}}_{\hA}\hB\|_F^2\right\}.
\end{equation}
In other words, we aim to find the optimal $\hA$, so that the
projection of the columns of $\h{\mathfrak{Y}}$ on the column space of
$\h{\mathfrak{X}}_{\hA}$ is maximized.

Taking partial derivatives of \eqref{eq:lse} with respect to the entries
of $\hA$ and $\hB$ respectively, we obtain the gradient condition for the LSE
\begin{equation}
  \label{eq:lse_grad}
\begin{aligned}
   \sum_t\hA\hX_{t-1}\hB'\hB\hX_{t-1}'-\sum_t\hX_t\hB\hX_{t-1}'&=\hzero\\
    \sum_t\hB\hX_{t-1}'\hA'\hA\hX_{t-1}-\sum_t\hX_t'\hA\hX_{t-1}&=\hzero.
\end{aligned}
\end{equation}
The function $\sum_{t}\|\hX_t-\hA\hX_{t-1}\hB'\|_F^2$ is guaranteed to
have at least one global minimum, so solutions to \eqref{eq:lse_grad}
are guaranteed to exist. On the other hand, if $\hat\hA$ and $\hat\hB$
solve the equations in \eqref{eq:lse_grad},
so are
$\tilde\hA:=\hat\hA/c$ and $\tilde\hB:=\hat\hB\cdot c$, where $c$ is
any nonzero constant. We should regard them as the same solution
because they yield the same matrix product
$\hat\hA\hX_{t-1}\hat\hB'=\tilde\hA\hX_{t-1}\tilde\hB'$. Equivalently,
we say that $(\hat\hA,\hat\hB)$ and $(\tilde\hA,\tilde\hB)$ are the
same solution of \eqref{eq:lse_grad}, if
$\hat\hB\otimes\hat\hA=\tilde\hB\otimes\tilde\hA$. With this
convention, we argue that with probability one, the global minimum of
\eqref{eq:lse} is unique. For this purpose, we need the following
condition.

\noindent
({\bf Condition R:}) \
The innovations $\hE_t$ are independent and
identically distributed, and absolutely continuous with respect to
Lebesque measure.


If Condition ({\bf R}) is fulfilled, it holds that with probability
one, the solutions of \eqref{eq:lse_grad} have full ranks, and they
have no zero entries. Let us restrict our discussion to this event of
probability one. Without loss of generality, we fix the first entry of
$\hA$ at $a_{11}=1$. Let us use $\mathcal{Z}$ to denote the set of
entries of $\hA$ and $\hB$:
$$\mathcal{Z}:=\{a_{ij},b_{kl}:\;1\leq i,j\leq m, (i,j)\neq(1,1), 1\leq
k,l\leq n\}.$$ The matrix equations in \eqref{eq:lse_grad}
involves
$m^2+n^2$ individual equations, and each equation takes the form
$f(\mathcal{Z})=0$, where $f(\mathcal{Z})$ is a multivariate
polynomial in the polynomial ring $\C[\mathcal{Z}]$ over the complex
field $\C$. The collection $\h{V}$ of all solutions of
\eqref{eq:lse_grad}
is thus an affine variety in the space
$\C^{m^2+n^2-1}$. By computing a Groebner basis for the ideal
generated by the polynomials in \eqref{eq:lse_grad},
we see that $\hV$
is a finite set \citep[see for example, Theorem~6, page 251
of][]{cox:2015}. Equivalently, the equations in \eqref{eq:lse_grad}
have a finite number of solutions.

Now consider the uniqueness of the global minimum. Assume $\hA_1$ and
$\hA_2$ are two nonzero $m\times m$ matrices such that
$\hA_1\neq c\cdot\hA_2$ for any constant $c\in\R$. Define
$\h{\mathfrak{X}}_1$ and $\h{\mathfrak{X}}_2$ as in \eqref{eq:lse_bla}
with $\hA_1$ and $\hA_2$ respectively. We further define
$\h{\mathfrak{X}}(c)=c\h{\mathfrak{X}}_1+(1-c)\h{\mathfrak{X}}_2$. The
projection of $\h{\mathfrak{Y}}$ on the column space of
$\h{\mathfrak{X}}(c)$ is given by
$\h{\mathfrak{X}}(c)[\h{\mathfrak{X}}(c)'\h{\mathfrak{X}}(c)]^{-1}\h{\mathfrak{X}}(c)'\h{\mathfrak{Y}}$,
and its Frobenius norm 
\begin{equation}
  \label{eq:frob1}
  \mathrm{trace}\left\{\h{\mathfrak{Y}}'\h{\mathfrak{X}}(c)[\h{\mathfrak{X}}(c)'\h{\mathfrak{X}}(c)]^{-1}\h{\mathfrak{X}}(c)'\h{\mathfrak{Y}}\right\}
\end{equation}
is a rational function of $c\in\R$ with random
coefficients determined by $\hX_1,\hX_2,\ldots,\hX_T$. With
probability one, this rational function takes distinct
values at different local minima. Combining this fact and the
preceding argument, we see that with probability one, the least
squares problem \eqref{eq:lse} has an unique global minimum, and
finitely many local minima.

To solve \eqref{eq:lse}, we iteratively
update the two matrices $\hat\hA$ and $\hat\hB$,
by updating one of them in the least squares
\eqref{eq:lse}
while holding the other one fixed, starting with some initial $\hA$ and $\hB$.
By \eqref{eq:lse_grad}
the iteration of updating $\hB$ given $\hA$ is
\begin{equation*}
  \hB \leftarrow \left(\sum_{t}\hX_t'\hA\hX_{t-1}\right)\left(\sum_t\hX_t'\hA'\hA\hX_{t-1}\right)^{-1},
\end{equation*}
and similarly by \eqref{eq:lse_grad}
the iteration of updating $\hA$ given $\hB$ is
\begin{equation*}
  \hA \leftarrow \left(\sum_t \hX_t\hB\hX_{t-1}'\right)\left(\sum_t\hX_{t-1}\hB'\hB\hX_{t-1}'\right)^{-1}.
\end{equation*}
We denote these estimators by $\hat\hA_2$ and $\hat\hB_2$, with the
name {\it least squares estimators}, and the acronym LSE.

The iterative least squares may converge to a local minimum. In practice,
we suggest to use the PROJ estimators $\hat\hA_1$ and $\hat\hB_1$ as the
starting values of the iterations. On the other hand, by permuting the
entries of the corresponding matrices \citep{vanloan:2000}, the problem
\eqref{eq:lse} can be rewritten as a problem of best rank-one
approximation under a linear transform, which in turn can be viewed as
a generalized SVD
problem. 
The variable projection methods discussed in \cite{golub:1973} and
\cite{kaufman:1975} may also be applicable
here. 

\subsection{MLE under a structured covariance tensor}

When the covariance matrix of the error matrix $\hE_t$ assumes the
structure in \eqref{eq:ocov}, it can be utilized to improve the
efficiency of the estimators. The log likelihood under normality can
be written as
\begin{equation}
\label{eq:loglik}
  -m(T-1)\log|\Sigma_c| - n(T-1)\log|\Sigma_r| - \sum_{t}\tr\left(\Sigma_r^{-1}(\hX_t - \hA\hX_{t-1}\hB')\Sigma_c^{-1}(\hX_t - \hA\hX_{t-1}\hB')'\right).
\end{equation}
Four matrix parameters $\hA,\hB,\Sigma_r,\Sigma_c$ are involved in the
log likelihood function. The gradient condition at the MLE is given by
\begin{equation*}
  \begin{aligned}
    \hA\sum_t\hX_{t-1}\hB'\Sigma_c^{-1}\hB\hX_{t-1}'-\sum_t\hX_t\Sigma_c^{-1}\hB \hX_{t-1}' & =\hzero \\
    \hB\sum_t\hX_{t-1}'\hA'\Sigma_r^{-1}\hA\hX_{t-1}-\sum_t\hX_t'\Sigma_r^{-1}\hA \hX_{t-1} & =\hzero \\
    m(T-1)\Sigma_c-\sum_t(\hX_t - \hA\hX_{t-1}\hB')'\Sigma_r^{-1}(\hX_t - \hA\hX_{t-1}\hB') & = \hzero \\
    n(T-1)\Sigma_r-\sum_t(\hX_t - \hA\hX_{t-1}\hB')\Sigma_c^{-1}(\hX_t - \hA\hX_{t-1}\hB')' & = \hzero.
  \end{aligned}
\end{equation*}
To find the MLE, we iteratively update one, while
keeping the other three fixed. These iterations are given by
\begin{align*}
  \hA &\leftarrow \left(\sum_t\hX_t\Sigma_c^{-1}\hB \hX_{t-1}'\right)
        \left(\sum_t\hX_{t-1}\hB'\Sigma_c^{-1}\hB\hX_{t-1}'\right)^{-1}\\
  \hB &\leftarrow \left(\sum_t\hX_t'\Sigma_r^{-1}\hA \hX_{t-1}\right)
        \left(\sum_t\hX_{t-1}'\hA'\Sigma_r^{-1}\hA\hX_{t-1}\right)^{-1}\\
  \Sigma_c &\leftarrow \frac{\sum_t \hR_t'\Sigma_r^{-1}\hR_t}{m(T-1)},
             \mbox{ where } \hR_t = \hX_t - \hA\hX_{t-1}\hB'\\
  \Sigma_r & \leftarrow \frac{\sum_t \hR_t\Sigma_c^{-1}\hR_t'}{n(T-1)},
             \mbox{ where } \hR_t = \hX_t - \hA\hX_{t-1}\hB'
\end{align*}
The MLE for $\hA$ and $\hB$ under the covariance structure
\eqref{eq:ocov} will be denoted by $\hat\hA_3$ and $\hat\hB_3$, with
an acronym {MLEs}, where the ``s'' emphasizes the fact that it is the MLE under the special
structure \eqref{eq:ocov} of the error covariance matrix. 
Due to the unidentifiability 
of the pairs of $\hA, ~\hB$ and $\Sigma_c,~\Sigma_r$, 
to make sure the numerical computation is stable, after looping through 
$\hA, ~\hB, ~\Sigma_c$ and $\Sigma_r$ in each iteration,
we renormalize so that both $||\hA||_F=1$ and $||\Sigma_r||_F=1$.

\noindent
{\bf Remark: }
Note that the three estimators do not impose the causality condition $\rho(\hA)\rho(\hB)<1$ in the estimation procedure. Hence the resulting estimators may not necessarily satisfy the condition, even the underlying process is stationary and causal. For unitvariate ARMA modeling, a transformation of the estimator can be made to achieve causality, and the transformed model and the original one are equivalent under Gaussianity \citep[see for example,][Section~3.5]{brockwell:1991}. There is a similar result for VAR models (see for example \cite{hannan:1970}, Section~II.5). Unfortunately the approach does not work in general under the restricted form of MAR(1) model, because the autoregressive coefficient matrix of the equivalent causal VAR(1) model no longer has the form of a Kronecker product. The hope is that if the process is indeed causal, the consistencies of the estimators (see Section 4) guarantee that they will satisfy the causality condition with large probabilities. On the other hand, to retain a MAR(1) model with possibly non-causal coefficient matrices, non-causal vector autoregression may be considered \citep{davis2012noncausal,lanne2013noncausal}.
    
\section{Asymptotics, Efficiency and a Specification Test}
\label{sec:asy}

\subsection{Asymptotics and Efficiency}

Due to the identifiability issue regarding $\hA$ and $\hB$, we make
the convention that $\|\hA\|_F=1$, and the three estimators
$(\hat \hA_i,\hat\hB_i),\;1\leq i\leq 3$ are rescaled so that
$\|\hat\hA_i\|_F=1$. Since the Kronecker product $\hB\otimes\hA$ is
unique, we also state the asymptotic distributions of the
estimated Kronecker product $\hB\otimes\hA$, in addition to that of
$\hA$ and $\hB$.

We first present the central limit theorem for the projection
estimators $\hat\hA_1$ and $\hat\hB_1$. Following standard theory of
multivariate ARMA models \citep{hannan:1970,dunsmuir:1976}, the
conditions of Theorem~\ref{thm:clt1} guarantees that $\hat\Phi$
converges to a multivariate normal distribution:
\nocite{deistler:1978}
\begin{equation*}
  \sqrt{T}\mathrm{vec}(\hat\Phi-\hB\otimes\hA) \Rightarrow
  N(\hzero,\Gamma_0^{-1}\otimes\Sigma),
\end{equation*}
where $\Sigma$ is the covariance matrix of $\vect(\hE_t)$, and
$\Gamma_0$ is given in \eqref{eq:autocov}. Let
$\tilde\Phi={\cal G}(\hat\Phi)$ be the rearranged version of
$\hat\Phi$, and $\Xi_1$ be the asymptotic covariance matrix of
$\vect(\tilde\Phi)$. The matrix $\Xi_1$ is obtained by {rearranging}
the entries of $\Gamma_0^{-1}\otimes\Sigma$, and can be expressed
using permutation matrices and Kronecker products, but we omit the
explicit formula here. 
\begin{theorem}
  \label{thm:clt1}
  Consider model \eqref{eq:mar1}. Set $\halpha:=\vect(\hA)$,
  $\hbeta_1:=\vect(\hB)/\|\vect(\hB)\|$, and
  \begin{align*}
    \hV_0&:=
           \begin{pmatrix}
             \|\hB\|_F^{-1}[\hbeta_1'\otimes(\hI-\halpha\halpha')] \\
             \hI\otimes\halpha'
           \end{pmatrix},\\
    \hV_1&:=(\hbeta_1\hbeta_1')\otimes\hI + \hI\otimes(\halpha\halpha')
    - (\hbeta_1\hbeta_1')\otimes(\halpha\halpha').
  \end{align*}
  Note that both $\halpha$ and $\hbeta_1$ are unit vectors. Assume that $\hE_1,\ldots,\hE_T$ are iid with mean zero and finite
  second moments. Also assume the causality condition
  $\rho(\hA)\cdot\rho(\hB)<1$, and $\hA$, $\hB$ and $\Sigma$ are
  nonsingular. It holds that
  \begin{equation*}
    \sqrt{T}\begin{pmatrix}
      \vect(\hat\hA_1-\hA) \\
      \vect(\hat\hB_1-\hB)
    \end{pmatrix} \Rightarrow N(\hzero,\hV_0\Xi_1\hV_0'),
  \end{equation*}
  and
  \begin{equation*}
    \sqrt{T}\left[\vect(\hat\hB_1)\otimes\vect(\hat\hA_1)-\vect(\hB)\otimes\vect(\hA)\right] \Rightarrow
    N(\hzero,\hV_1\Xi_1\hV_1').
  \end{equation*}
\end{theorem}

\noindent The proof of the theorem is presented in Appendix.

Note that although the projection estimator does not utilize the MAR(1) model structure, 
Theorem 2 requires that the observed matrix time series follows model (\ref{eq:mar1}).


Now we consider the least squares estimators
$(\hat\hA_2,\hat\hB_2)$. Let $\halpha:=\vect(\hA)$,
$\hbeta:=\vect(\hB')$, 
and $\hgamma:=(\halpha',\hzero')'$ be a vector
in $\R^{m^2+n^2}$. Note that $\hbeta$ should not be confused with
$\hbeta_1$ defined in Theorem~\ref{thm:clt1}. In Theorem~\ref{thm:clt1} we present the result for $\vect(\hB)$, of which $\hbeta_1$ is the normalized version; while here in Theorem~\ref{thm:lse}, we give CLT for $\vect(\hB')$, which is denoted by $\hbeta$. Recall that $\Sigma$ is the covariance matrix of $\vect(\hE_t)$. We have the following result for $\hat\hA_2$ and $\hat\hB_2$.
\begin{theorem}
  \label{thm:lse}
  Consider model \eqref{eq:mar1}. Define $\hW_t':=[(\hB\hX_t')\otimes\hI:\hI\otimes(\hA\hX_t)]$,
  and $\hH:=\E(\hW_t\hW_t')+\hgamma\hgamma'$. Let
  $\Xi_2 := \hH^{-1}\E(\hW_t\Sigma\hW_t')\hH^{-1}$, and
  $\hV:=[\hbeta\otimes\hI,\hI\otimes\halpha]$. In addition to the
  conditions of Theorem~\ref{thm:clt1}, we also assume (R). It holds
  that
  \begin{equation*}
    \sqrt{T}\begin{pmatrix}
      \vect(\hat\hA_2-\hA) \\
      \vect(\hat\hB_2'-\hB')
    \end{pmatrix}
    \Rightarrow N(\hzero, \Xi_2);
  \end{equation*}
  and equivalently,
  \begin{equation*}
    \sqrt{T}\left[\vect(\hat{\hB}_2')\otimes\vect(\hat{\hA}_2)-\vect(\hB')\otimes\vect(\hA)\right] \Rightarrow
    N(0,\hV\Xi_2\hV').
  \end{equation*}
\end{theorem}

\noindent The proof of the theorem is in Appendix.

With the additional assumption \eqref{eq:ocov} on the covariance
structure of $\hE_t$, we have a similar result. Recall that
$\hW_t'=[(\hB\hX_t')\otimes\hI,\hI\otimes(\hA\hX_t)]$, and
$\hgamma=(\halpha',\hzero')'$. Let
$\tilde\hH:=\E(\hW_t\Sigma^{-1}\hW_t')+\hgamma\hgamma'$. 
Define
$\Xi_3 :=
\tilde\hH^{-1}\E(\hW_t\Sigma^{-1}\hW_t')\tilde\hH^{-1}$. Note that the
covariance matrix $\Sigma$ takes the form
$\Sigma=\Sigma_c\otimes\Sigma_r$. The MLEs $\hat\hA_3$ and $\hat\hB_3$
under the assumption \eqref{eq:ocov} have the following joint limiting
distribution.
\begin{theorem}
  \label{thm:mle}
  Under the same conditions of Theorem~\ref{thm:lse}, and the additional assumption \eqref{eq:ocov}, it holds
  that
  \begin{equation*}
    \sqrt{T}\begin{pmatrix}
      \vect(\hat\hA_3-\hA) \\
      \vect(\hat\hB_3'-\hB')
    \end{pmatrix}
    \Rightarrow N(\hzero, \Xi_3);
  \end{equation*}
  and equivalently,
  \begin{equation*}
    \sqrt{T}\left[\vect(\hat{\hB}_3')\otimes\vect(\hat{\hA}_3)-\vect(\hB')\otimes\vect(\hA)\right] \Rightarrow
    N(0,\hV\Xi_3\hV').
  \end{equation*}
\end{theorem}

\noindent The proof of the theorem is in Appendix.

We remark that in these three versions of the central limit theorem,
there may be zero diagonal entries in the asymptotic covariance
matrix. For example, in Theorem~\ref{thm:lse}, there may be a zero on
the diagonal of the matrix $\hV\Xi_2\hV'$. It happens when the
corresponding true values of the entries $a_{i_1j_1}$ and $b_{i_2j_2}$
are both zero; and in this situation, the product estimator
$\hat a_{i_1j_1}\hat b_{i_2j_2}$ has a convergence rate of $1/T$
instead of $1/\sqrt{T}$.

We now compare the efficiencies of the LSE and the MLEs, when the
covariance matrix of $\vect(\hE_t)$ has the Kronecker product
structure \eqref{eq:ocov}. In Theorems~\ref{thm:lse} and
\ref{thm:mle}, both asymptotic covariance matrices take the
form $\hV\Xi_i\hV',\;i=2,3$. The following corollary asserts that the
MLEs is asymptotically more efficient under the structured covariance
matrix \eqref{eq:ocov}.
\begin{corollary}
  \label{thm:eff}
  Consider model \eqref{eq:mar1}, and assume the same conditions of Theorem~\ref{thm:mle}, 
  It holds that
  \begin{equation*}
    \Xi_2 \geq \Xi_3.
  \end{equation*}
\end{corollary}
\noindent The proof of the Corollary is in Appendix.

Here the matrix relationship $\geq$ means that the difference of the
two matrices is positive semi-definite. Consequently, we see that when
the covariance structure is correctly specified by \eqref{eq:ocov},
the MLEs $\hat\hA_3$ and $\hat\hB_3$ are more efficient than the LSE
$\hat\hA_2$ and $\hat\hB_2$ asymptotically. A comparison of the
efficiencies of the projection estimators and least squares estimators
can also be made, where the least squares estimators are more
efficient. However, we skip this result here, because in practice we
suggest to use either LSE or MLEs, and only use PROJ as initial values
for the other two estimation methods.

\subsection{A Specification Test}

To assess the adequacy of the MAR(1) for a given dataset, it is natural to run some diagnostics based on the residuals. Since the MAR(1) model can be viewed as a special case of the VAR(1) model, standard diagnostics can be applied. Autocorrelation and cross correlation plots are useful to visualize the whiteness of the residual matrices. Portmanteau tests \citep{hosking:1980,hosking:1981,li:1981,poskitt:1982}, Lagrange multiplier test \citep{hosking:1981lm}, and the likelihood ratio test \citep{tiao:1981} can all be applied to test for serial correlations among the residual matrices. 

On the other hand, the fact that the MAR(1) model is a VAR(1) of a special form also makes it interesting to compare MAR(1) with the unrestricted VAR(1), and to examine whether the special form \eqref{eq:ovar1} is supported by the data. We propose a specification test based on the projection estimators $\hat\hA_1,\hat\hB_1$. We first state a corollary of Theorem~\ref{thm:clt1}, which will motivate the test statistic. The proof will be deferred to Appendix. Let $\h{M}^+$ be the Moore-Penrose inverse of a matrix $\h{M}$. Recall that in Theorem~\ref{thm:clt1}, we define $\halpha:=\vect(\hA)$ and $\hbeta_1:=\vect(\hB)/\|B\|_F$. 
Define the orthogonal projection matrix $\hP:=(\hI-\hbeta_1\hbeta_1')\otimes(\hI-\halpha\halpha')$. Note that $\hP=\hI-\hV_1$, where $\hV_1$ is defined in Theorem~\ref{thm:clt1}. 

\begin{corollary}
\label{thm:test}
  Assume the same conditions, and adopt the same notations of Theorem~\ref{thm:clt1}. Let
  \begin{equation*}
      \hat\hD:=\left[\tilde\Phi-\vect(\hat\hA_1)\vect(\hat\hB_1)'\right]
  \end{equation*}
  It holds that
  \begin{equation*}
      T\cdot\vect(\hat\hD)'\;\left(\hP\,\Xi_1\,\hP\right)^+\;\vect(\hat\hD)\Rightarrow\chi^2_{(m^2-1)(n^2-1)}.
  \end{equation*}
\end{corollary}

Recall the notations introduced before Theorem~\ref{thm:clt1}: $\Xi_1$ is the asymptotic covariance matrix of $\vect(\tilde\Phi)$, where $\tilde\Phi={\cal G}(\hat\Phi)$ is the rearranged version of $\hat\Phi$. The matrix $\Xi_1$ is obtained by {rearranging} the entries of $\Gamma_0^{-1}\otimes\Sigma$. Note that both $\Gamma_0$ and $\Sigma$ can be estimated by their sample versions. We denote by $\hat\Xi_1$ the corresponding estimator of the asymptotic covariance matrix of $\vect(\tilde\Phi)$. On the other hand, if $\halpha$ and $\hbeta_1$ are in the matrix $\hP$ are substituted by $\vect(\hat\hA_1)$ (note that we have made the convention that $\|\hat\hA_1\|_F=1$)
and $\vect(\hat\hB_1)/\|\hat\hB_1\|_F$ respectively, we have the estimator $\hat\hP$ for $\hP$. We consider the VAR(1) model \eqref{eq:var1} for $\vect(\hX_t)$, and test the hypothesis:
\begin{equation*}
    {H}_0:\; \Phi\hbox{ takes the form }\hB\otimes\hA \quad\hbox{vs}\quad
    {H}_1:\; \Phi\hbox{ cannot be expressed as } \hB\otimes\hA.
\end{equation*}
Motivated by Corollary~\ref{thm:test}, we use the test statistic
\begin{equation*}
    T\cdot\vect(\hat\hD)'\;(\hat\hP\,\hat\Xi_1\hat\hP)^{+}\;\vect(\hat\hD).
\end{equation*}
As an immediate consequence of Corollary~\ref{thm:test}, the test statistic also has the limiting distribution $\chi^2_{(m^2-1)(n^2-1)}$, based on which we are able to calculate the $p$-value.

\section{Numerical Results}
\label{sec:num}

\subsection{Simulations}

\begin{figure}[h]
  \centering
  \includegraphics[width=\textwidth]{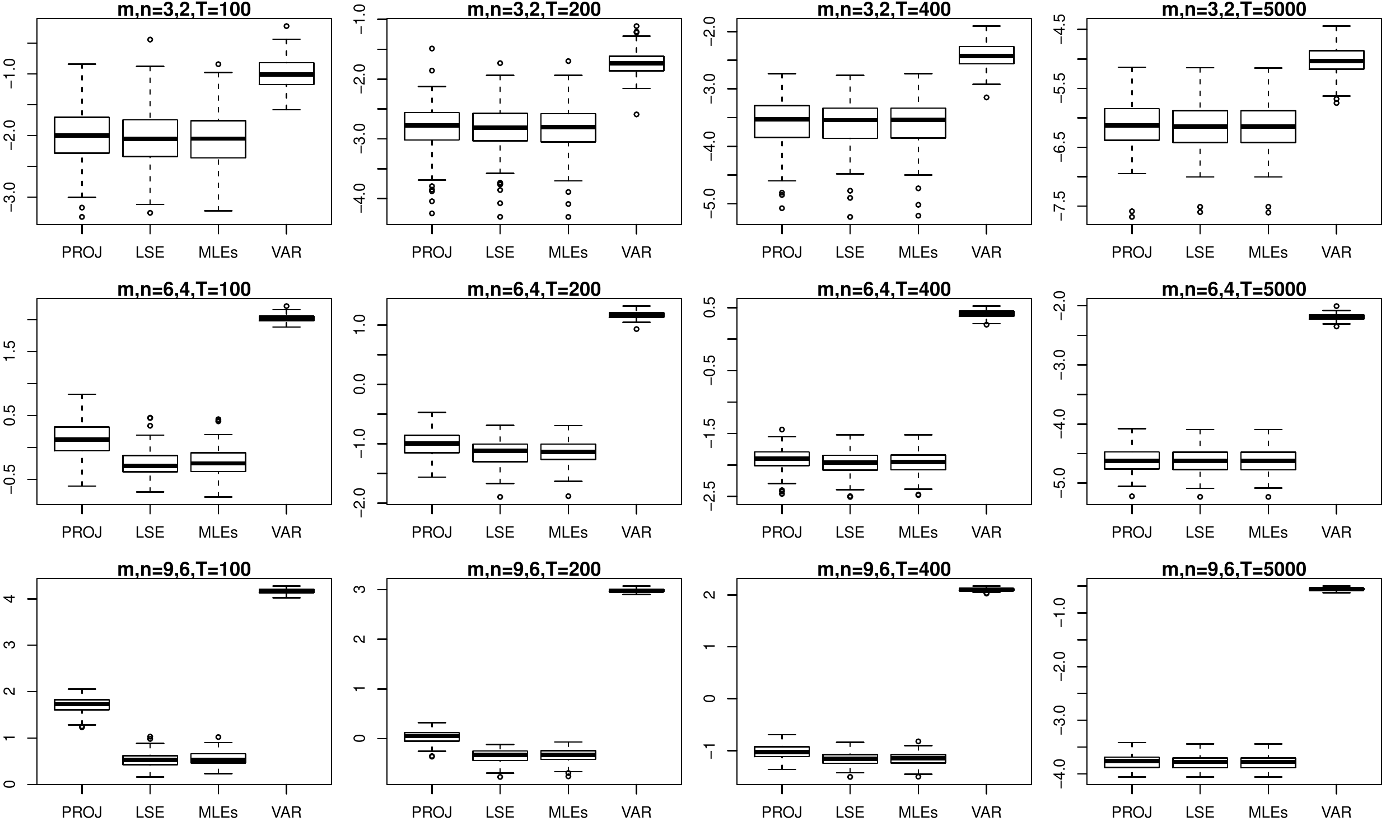}
  \caption{Comparison of four estimators,
PROJ, LSE, MLEs, and VAR, under Setting I. The three rows correspond to $(m,n)=(3,2), ~(6,4)$ and $(9,6)$ respectively, and the four columns $T=$ 100, 200, 400 and 5000 respectively.}
  \label{fig:sim1}
\end{figure}

\begin{figure}[h]
  \centering
  \includegraphics[width=\textwidth]{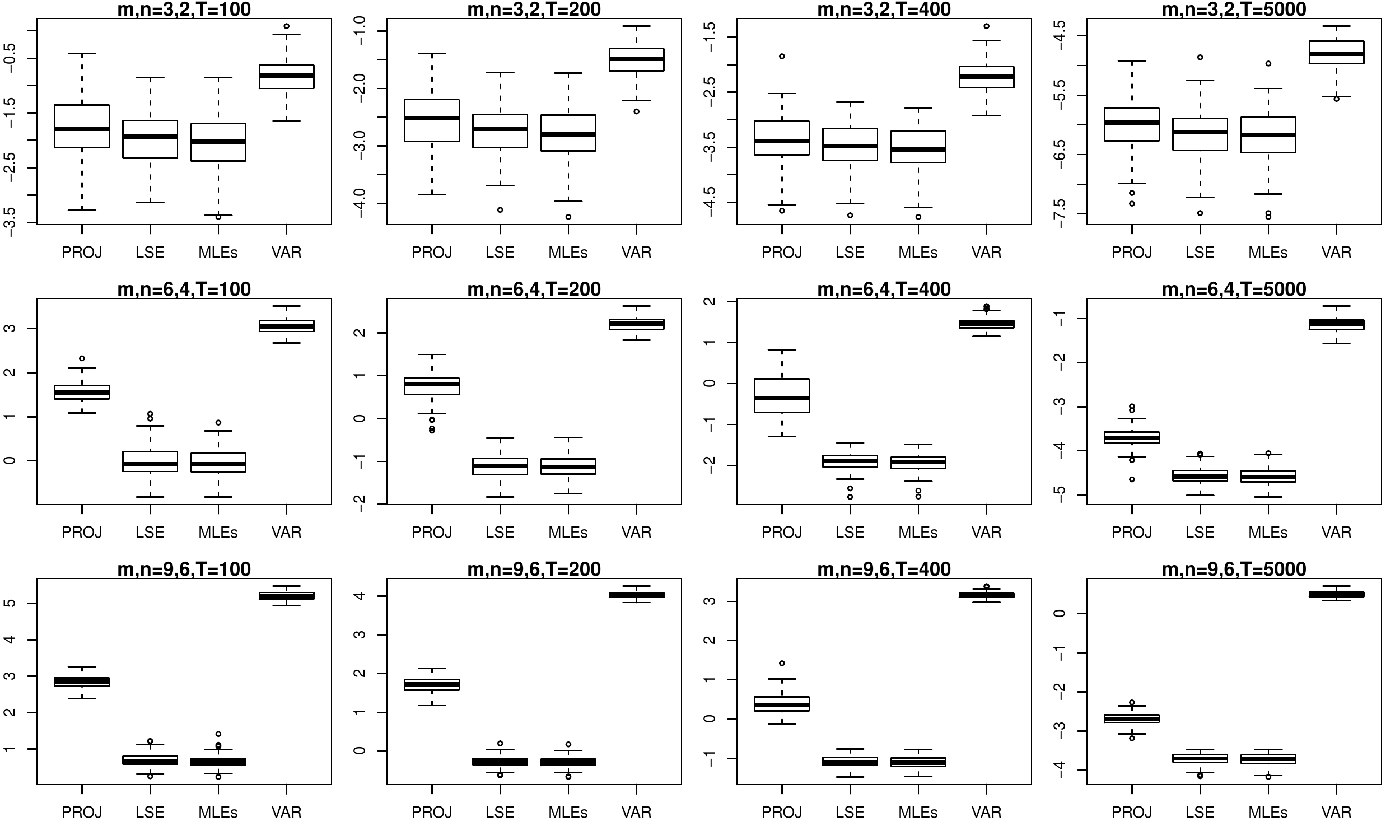}
  \caption{Comparison of four estimators, PROJ, LSE, MLEs, and VAR, under Setting II. The three rows correspond to $(m,n)=(3,2), ~(6,4)$ and $(9,6)$ respectively, and the four columns $T=$ 100, 200, 400 and 5000 respectively.}
  \label{fig:sim2}
\end{figure}

\begin{figure}[h]
  \centering
  \includegraphics[width=\textwidth]{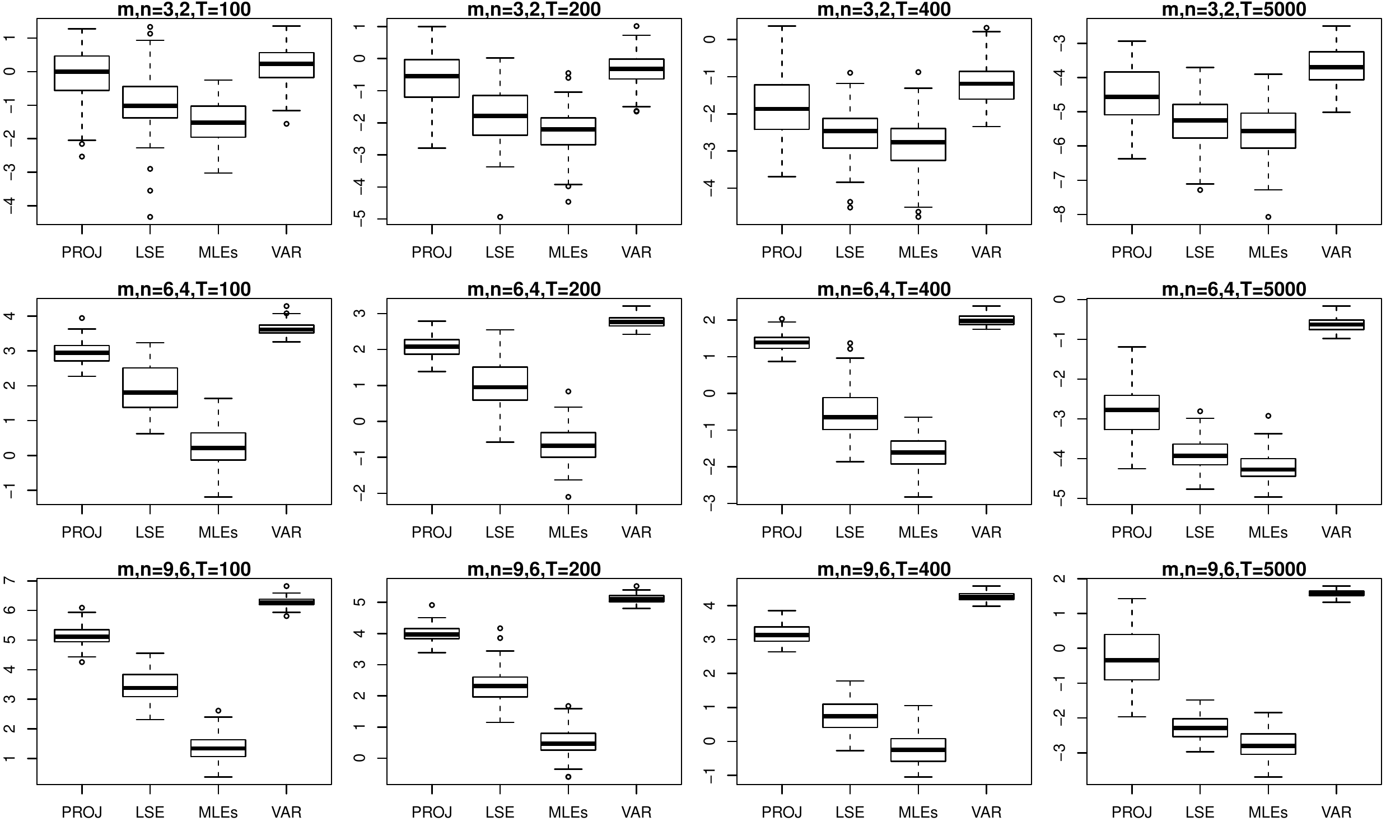}
  \caption{Comparison of four estimators, PROJ, LSE, MLEs, and VAR, under Setting III. The three rows correspond to $(m,n)=(3,2), ~(6,4)$ and $(9,6)$ respectively, and the four columns $T=$ 100, 200, 400 and 5000 respectively.}
  \label{fig:sim3}
\end{figure}

\begin{figure}[h]
  \centering
  \includegraphics[width=\textwidth]{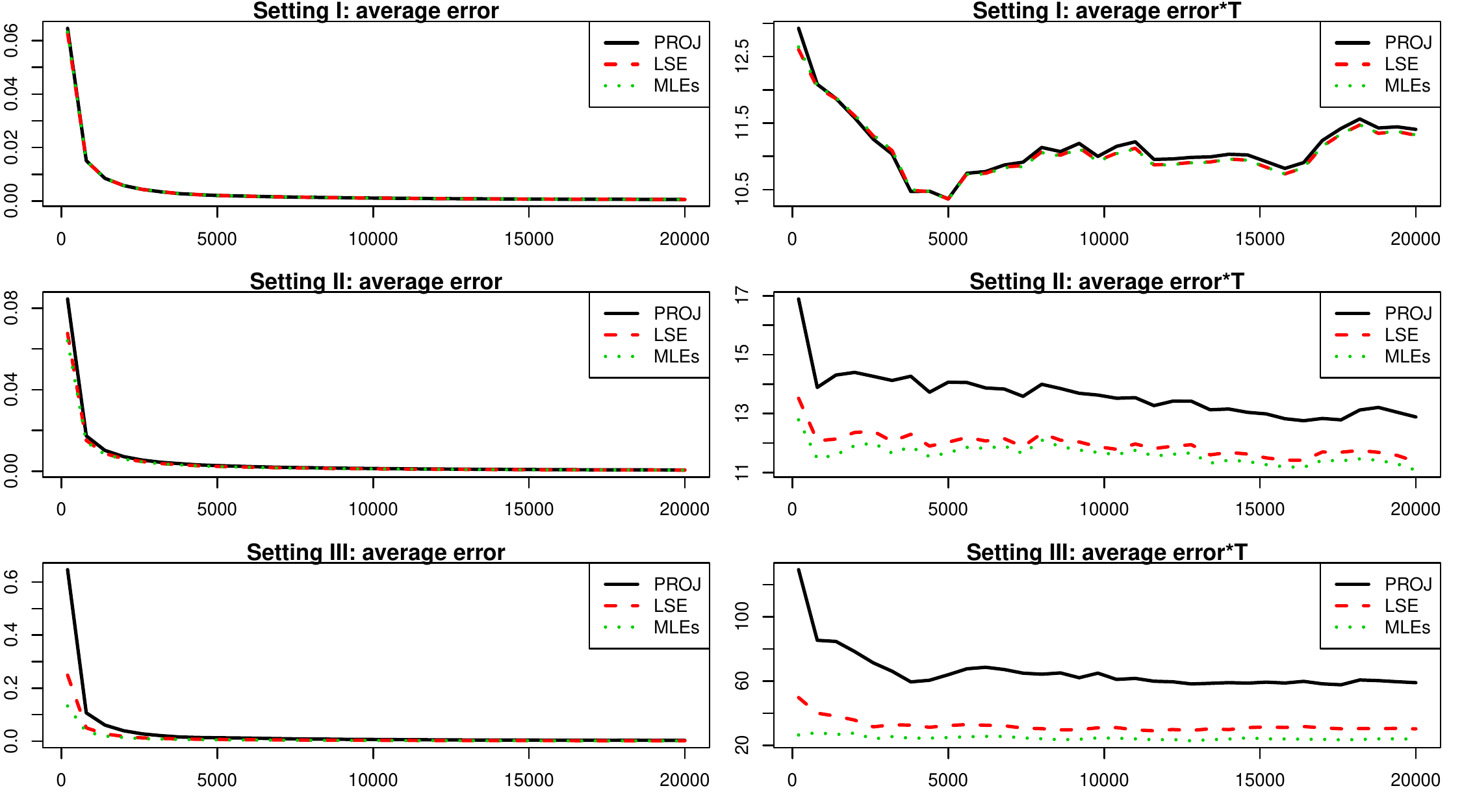}
  \caption{Comparison of the asymptotic efficiencies of three estimators, PROJ, LSE, and MLEs, under three settings. The left column shows the average error over 100 repetitions for $\|\hat\hB\otimes\hat\hA-\hB\otimes\hA\|_F^2$ and the right for $T\|\hat\hB\otimes\hat\hA-\hB\otimes\hA\|_F^2$. 
  }
  \label{fig:sim4} 
\end{figure}


In this section, we compare the performances of the aforementioned
estimators and the stacked VAR(1) estimator under different settings
for various choices of the matrix dimensions $m$ and $n$, as well as
the length of the time series $T$.

Specifically, for given dimensions $m$ and $n$, the observed
data $\hX_t$ are simulated according to model \eqref{eq:mar1}, where the
entries of $\hA$ and $\hB$ are generated randomly and then rescaled so
that $\rho(\hA)\rho(\hB)=.5$ to
guarantee the fulfillment of the causality condition and the constraint
$||\hA||_F=1$. For a
particular simulation setting with multiple repetitions, the
coefficient matrices $\hA$ and $\hB$ remain fixed.

In what follows, we perform six experiments: the first three
experiments demonstrate the finite-sample comparisons under three
settings of the covariance structure of the innovation matrix $\hE_t$
respectively, the fourth one compares the asymptotic properties of
all estimators when $T\rightarrow\infty$ under these three settings, 
the fifth one studies the finite-sample behavior of the asymptotic 
variance of the estimators, and the sixth one investigates the performance
of the specification test.

\begin{itemize}
\item Setting I: The covariance matrix $\Cov(\vect(\hE_t))$ is set to
$\Sigma=\mathbf{I}$. 

\item Setting II: The covariance matrix $\Cov(\vect(\hE_t))=\Sigma$ is
  randomly generated according to
  $\Cov(\vect(\hE_t))=\mathbf{Q}\Lambda\mathbf{Q}'$, where the
  eigenvalues in the diagonal matrix $\Lambda$ are the absolute values
  of i.i.d. standard normal random variates, and the eigenvector matrix
  $\mathbf{Q}$ is a random orthonormal matrix.
\item Setting III: The covariance matrix $\Cov(\vect(\hE_t))$ takes
  the Kronecker product form \eqref{eq:ocov}, where $\Sigma_c$ and
  $\Sigma_r$ are generated similarly as the $\Sigma$ in Setting II.
\end{itemize}

In addition to
the three estimators (PROJ, LSE, and MLEs)
discussed in Section~\ref{sec:est},
we also include the MLE under the stacked VAR(1)
model in \eqref{eq:var1}, with the acronym VAR, as a benchmark for comparison. For each
configuration, we repeat the simulation 100 times, and show a box plot
of
\begin{equation*}
  \log(\|\hat\hB\otimes\hat\hA-\hB\otimes\hA\|_F^2).
\end{equation*}

Figures~\ref{fig:sim1} to \ref{fig:sim3} show the simulation results
under three settings respectively for relatively small sample
sizes. For each of these three figures, the dimensions $m,n$ increase
from top to bottom, taking values in $(m,n)=(3,2), ~(6,4)$ and
$(9,6)$. The sample size $T$ increases from left to right at
$T=100, 200, 400$ and $5000$, respectively.
One common finding from these three
figures is that all three estimators, PROJ, LSE, and MLEs, obtained
under the MAR(1) model in \eqref{eq:mar1} outperform the stacked VAR
estimator.

In the first experiment under Setting I, Figure~\ref{fig:sim1} shows
that LSE is the best estimator when the covariance matrix is
indeed a diagonal matrix. This is intuitive since LSE is the maximum
likelihood
estimator under this setting. The very close second best is the MLEs,
which is comparable with LSE throughout different combinations of
$m,~n,~T$ and only performs slightly worse when $m,~n$ are large and
$T$ is small. This is expected since MLEs has to estimate the additional
row and column covariance matrices of sizes $m\times m$ and $n\times
n$ when it is not necessary.
Both LSE and MLEs are superior over the PROJ, especially when the
sample size is small and the dimensions are large as seen in the lower
left corner of the figure. As sample size increases, the advantage
becomes less obvious.

In the second experiment under Setting II,
the overall pattern of Figure~\ref{fig:sim2} is similar to that in
Figure~\ref{fig:sim1}. The VAR estimator performs the worst;
PROJ is the second worst, and LSE and MLEs are very much similar.
Note that under Setting II, the covariance structure is arbitrary and
does not follow the Kronecker structure, which is the underlying assumption
for the MLEs. The LSE does not assume any covariance structure in the
 estimation process. Hence one would expect MLEs, obtained under the wrong
assumption, should perform worse than LSE. However,
the simulation results show that they perform similarly.

In the third experiment under Setting III, Figure~\ref{fig:sim3} shows
that MLEs dominates LSE for any choice of $m,~n,~T$, as
Corollary \ref{thm:eff} predicts. LSE in turn prevails PROJ, which in turn
always leads the stacked VAR estimator.

In the fourth experiment, we compare the asymptotic efficiencies of
PROJ, LSE, and MLEs by letting the length of the time series $T$ go to
infinity. The main purpose of this experiment is to obtain qualitative
understanding of the asymptotic covariances of different estimators.
Although Theorems \ref{thm:clt1} to \ref{thm:mle} provide the
theoretical form of the asymptotic covariances and Corollary
\ref{thm:eff} ascertains the relative magnitude of the errors from LSE
and MLEs under Setting III, we have little concrete insight on the
relative performances of the three estimators under other
settings. For this purpose, we fix the dimensions $(m,n)=(3,2)$ for
all three settings in this experiment.  Figure~\ref{fig:sim4} shows
the results. The left three panels show the average estimation errors
over 100 repetitions of $\|\hat\hB\otimes\hat\hA-\hB\otimes\hA\|_F^2$
for different $T$.  The right panels shows
$T\|\hat\hB\otimes\hat\hA-\hB\otimes\hA\|_F^2$ as a function of
$T$. The three rows correspond to the three settings respectively. In
each of the six panels, the solid line, dashed line, and dotted line
correspond to PROJ, LSE, and MLEs respectively.  The three figures on
the left show the decreasing trend of all three estimators as $T$
grows. The three figures on the right magnify the differences of the
three estimators, with a clear ordering.  PROJ estimator clearly has
the lowest efficiency.  Under Setting 1 (top panels), LSE and MLEs
performs similarly, since LSE is the maximum likelihood estimators under
the setting. MLEs estimates in total 4 more parameters in $\Sigma_r$ and
$\Sigma_c$. The bottom panels in the figure show that MLEs is more
efficient than LSE under Setting III, which is expected by
Corollary~\ref{thm:eff}. However, under Setting II (the middle
panels), it is interesting to observe that MLEs outperforms LSE
slightly but consistently, although MLEs is obtained under a wrong
model assumption.  This is probably due to the use of a regularized
covariance structure \eqref{eq:ocov}, which is beneficial because of
the dimension reduction from 21 parameters in an arbitrary
$\Sigma$ to $8$ in $\Sigma_c\otimes\Sigma_r$. Note that the performance 
also depends on how close the Kronecker product approximation is to 
the arbitrary (random) covariance matrix used in the simulation.

In the fifth experiment, the finite-sample performance of the asymptotic 
covariance matrices is demonstrated. We fix the dimensions to be $m=3$ 
and $n=2$, and the results are similar for larger dimensions. 
Under each of the three aforementioned settings,  
combining the three estimators, PROJ, LSE, and MLEs, and their
corresponding standard errors from Theorems \ref{thm:clt1} to \ref{thm:mle},
we create 95\% confidence interval of each parameter
based on the asymptotic normality distribution. 
In particular, two types of confidence intervals are constructed:
one for the entries of the matrices $\hA$ and $\hB$ separately, and
the other for the entries of $\vect(\hB)\otimes \vect(\hA)$. 
We repeat the experiment 1000 times.
Table~\ref{table:SE} shows the percentage that the true parameter 
falls within the marginal 95\% confidence interval of each parameter
for the three different estimators, under three different settings and different 
sample sizes. It can be seen from the table that the
coverage is quite accurate, especially in large sample cases. 
The properties for other nominal confidence levels, for example, 90\% and 99\%,
are similar in nature.

\begin{table}[h]
\centering
{\small 
\begin{tabular}{|c|c|ccc|ccc|ccc|}\hline
      &Setting&&I&&    &II&&    &III&\\\hline
      &Estimator& PROJ&  LSE& MLEs& PROJ&  LSE& MLEs& PROJ&  LSE& MLEs\\\hline
&T=100 &0.926&0.934&0.932&0.913&0.935&0.923&0.872&0.906&0.947\\
$(\vect'(\hat\hA),\vect'(\hat\hB))'$&T=200 &0.938&0.941&0.941&0.937&0.944&0.932&0.915&0.934&0.950\\
&T=1000&0.950&0.951&0.951&0.947&0.947&0.933&0.946&0.949&0.953\\\hline
&T=100 &0.915&0.923&0.921&0.905&0.922&0.911&0.860&0.885&0.936\\
$\vect(\hat\hB)\otimes\vect(\hat\hA)$&T=200 &0.935&0.938&0.937&0.930&0.939&0.928&0.903&0.923&0.945\\
&T=1000&0.950&0.952&0.951&0.946&0.945&0.932&0.942&0.944&0.950\\\hline
\end{tabular}
}
\caption{Percentage of coverages of 95\% confidence intervals.}
\label{table:SE}
\end{table}

In the sixth experiment, the performance of the specification test in Corollary \ref{thm:test} is investigated. To that end, the samples are generated according to the following models
\[
  \h{X}_t = .5\h{A}_1\h{X}_{t-1}\h{B}_1' + .5\eta\h{A}_2\h{X}_{t-1}\h{B}_2' + \h{E}_t,
\]
where $\rho(\hA_1)=\rho(\hB_1)=\rho(\hA_2)=\rho(\hB_2)=1, ~\eta=0,~0.05,~0.10,~0.15,~...~,~0.50$. When $\eta=0$, the null hypothesis is valid and when $\eta=0.05,~0.10,~...~,~0.50$, the alternative is true. The larger the value of $\eta$ is, the more severe the deviation from the null hypothesis. Again, we fix the dimensions to be $m=3$ and $n=2$, and the results are similar for larger dimensions. The significance level is set to be 0.05 and we perform the specification test for 10,000 replications of the data with five choices of length $T$ in each of the three aforementioned settings. Figure \ref{fig:sim6} shows the empirical sizes and powers as a function of $\eta$. It can be seen that the when $\eta=0$, the empirical sizes are close to 0.05 and the powers increase to 1 as $\eta$ increases under all three settings. It is also shown that as $T$ increases, the powers increase from 0 to 1 more quickly as a function of $\eta$.

\begin{figure}[h]
  \centering
  \includegraphics[width=\textwidth]{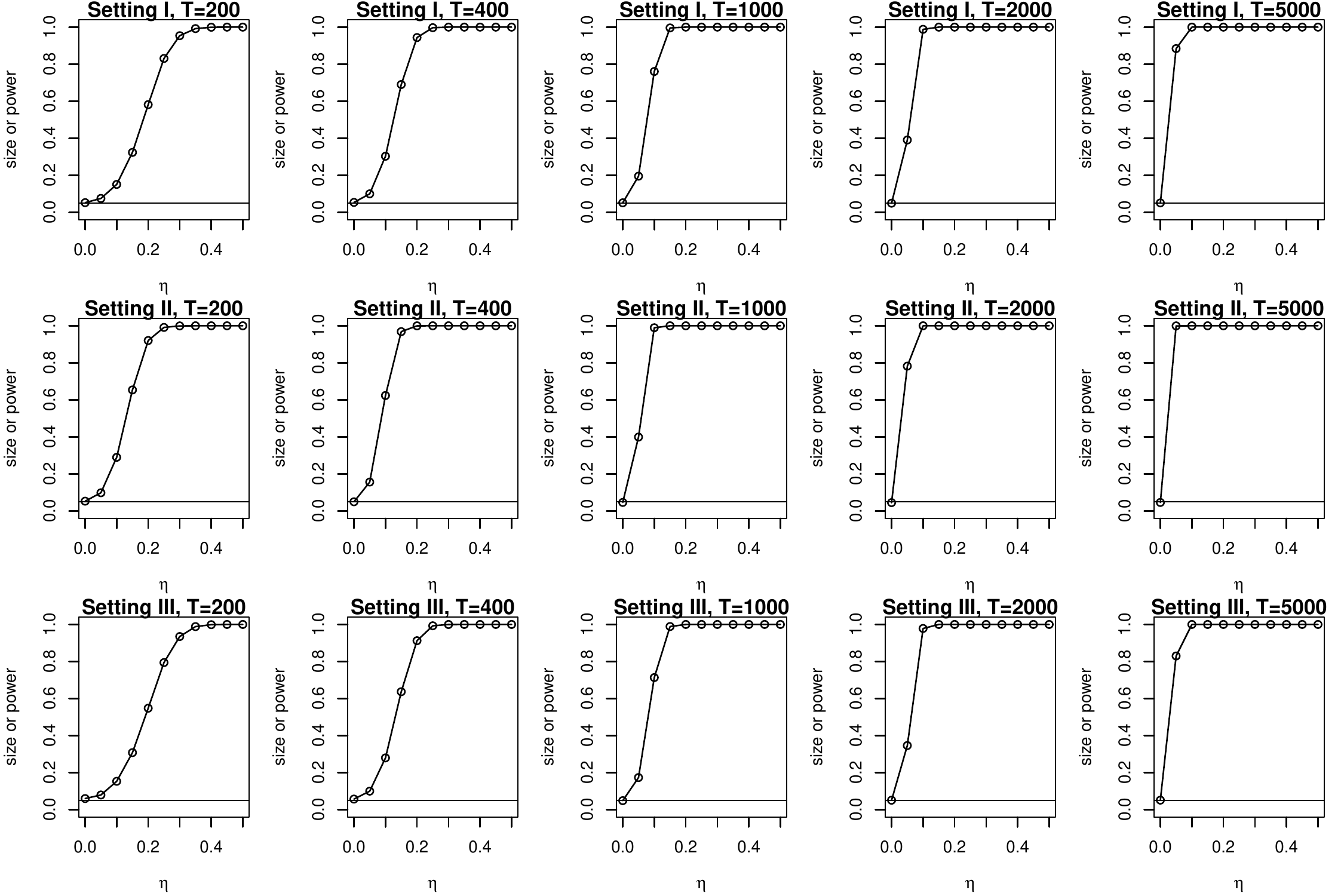}
  \caption{Power of the specification test for varying time series lengths in three settings. $\eta$ is a measure of how far the alternative hypothesis is away from the null hypothesis. The heights of the horizontal lines are 0.05.}
  \label{fig:sim6} 
\end{figure}

\subsection{Economic indicator from five countries}

We now revisit the example shown in Figure~\ref{fig:1}. The data
consists of quarterly observations of four economic indicators:
3-month interbank interest rate (first order differenced series), 
GDP growth (first order differenced log of GDP series), Total manufacture Production growth 
(first order differenced log of Production series) and total consumer price index (growth from
the last period) from five countries: Canada, France, Germany, United Kingdom and
United States. It ranges from 1990 to 2016. The data was obtained from
Organisation for Economic Co-operation and Development (OECD) at {\tt
  https://data.oecd.org/}.  Before fitting the autoregressive models,
we adjusted the seasonality of CPI by subtracting the sample quarterly
means. All series are normalized so that the combined variance of each
indicator (each row) is 1.

MAR(1) model was estimated using the three estimation methods. We also fitted
a stacked VAR(1) model, and
univariate AR(1) and AR(2) models for each individual time series. The residual
sum of squares of each model and the sum of squares of the (normalized)
original data are listed in Table~\ref{table:SS}. The MAR(1) estimated
using the least squares method has the smallest residual sum of squares,
among all models and methods,
except the VAR(1) model.  Note
that MAR(1) model uses $16+25-1=40$ parameters in the two coefficient
matrices, comparing to 20 and 40 parameters in
fitting 20
univariate AR(1) and AR(2) models to each series, respectively. The VAR(1)
model has total 400 parameters in the AR coefficient matrix. The large
number of parameters results in a small residual sum of squares. It
is deemed to be overfitting as we will show later in out-sample rolling
forecasting performance evaluation.

\begin{table}[h]
\centering
\begin{tabular}{ccc|c|cc|c}\hline
  MAR(1) PROJ & MAR(1) LSE & MAR(1) MLEs & VAR(1) & iAR(1) & iAR(2) & original \\ \hline
 1572 & 1468 & 1487 & 1121& 1724 & 1663 & 2116 \\ \hline
\end{tabular}
\caption{Residual sum of squares of MAR(1) model using three different
estimators and the stacked VAR(1) estimator;
and the total residual sum of squares of fitting
univariate AR(1) and AR(2) to each individual time series; and
the total sum of squares of the original (normalized) data.}
\label{table:SS}
\end{table}

Tables~\ref{table:LL} and \ref{table:RR} show the estimated parameters
and their corresponding standard errors (in the parentheses) of $\hA$
and $\hB$ using the least squares method. Due to ambiguity between the
two matrices, the left matrix is scaled so that its Frobenius norm is
one. On the right of the
table we also indicate the positively significant, negatively
significant and insignificant parameters (at 5\% level) using symbols
$(+,-,0)$, respectively.

\begin{table}[h]
\centering
\begin{tabular}{c|cccc||cccc}
 &      Int &    GDP &   Prod &    CPI &  Int &    GDP &   Prod &    CPI
 \\ \hline
Int &0.272  &0.304 &0.066 &0.018 &$+$ &$+$ &$0$ &$0$\\
    &(0.063) &(0.079) &(0.097) &(0.05)\\
GDP &-0.164 &0.447 &0.32  &-0.036 &$-$ &$+$ &$+$ &$0$\\
    &(0.052) &(0.077) &(0.094) &(0.045)\\
Prod &-0.198 &0.393 &0.429 &0.009 &$-$ &$+$ &$+$ &$0$\\
     &(0.057) &(0.083) &(0.101) &(0.05)\\
CPI &-0.108 &0.031 &0.036 &0.327  &$0$ &$0$ &$0$ &$+$\\
    &(0.072) &(0.106) &(0.118) &(0.059)\\
\end{tabular}
\caption{Estimated left coefficient matrix $\hA$ of MAR(1) using LS method.
Standard errors are shown in the parentheses. The right panel indicates
the positively significant, negatively significant and
insignificant  parameters at 5\% level using symbols $(+,-,0)$, respectively.}
\label{table:LL}
\end{table}

The left coefficient matrix shows an interesting pattern. 
For example, the first column in Table~\ref{table:LL} shows the influence on the current economic indicators from the past quarter's interest rate. The influence on the current GDP growth, Production growth and CPI are all negative, meaning that a higher interest rate will make the GDP growth and Production growth slower. Current CPI is also negatively related to a higher past interest rate.
The second column in Table~\ref{table:LL} shows that the influence on the current economic indicators from the past quarter's GDP growth. They are all positive, except the insignificant influence on CPI. 
The last row of Table~\ref{table:LL} shows that the past economic
indicators do not have significant influence on the current
CPI, except its own past; whilst the last column indicates that the past CPI does not have signifiant influence on all current indicators, except itself.

\begin{table}[h]
\centering
\begin{tabular}{c|ccccc||ccccc}
&    USA &  DEU &  FRA  &  GBR &  CAN & USA &  DEU &  FRA  &  GBR & CAN \\
\hline
USA &0.753 &-0.132 &0.159 &0.462 &-0.057 &$+$ &$0$ &$0$ &$+$ &$0$\\
    &(0.128) &(0.182) &(0.126) &(0.121) &(0.146)\\
DEU &0.45  &0.179  &0.678 &0.359 &-0.387 &$+$ &$0$ &$+$ &$+$ &$-$\\
    &(0.083) &(0.131) &(0.085) &(0.079) &(0.096)\\
FRA &0.363 &0.073  &0.292 &0.254 &0.041 &$+$ &$0$ &$+$ &$+$ &$0$\\
    &(0.128) &(0.198) &(0.139) &(0.126) &(0.154)\\
GBR &0.385 &-0.098 &0.099 &0.686 &0.069 &$+$ &$0$ &$0$ &$+$ &$0$\\
    &(0.115) &(0.159) &(0.109) &(0.115) &(0.132)\\
CAN &0.511 &-0.083 &0.054 &0.634 &0.308 &$+$ &$0$ &$0$ &$+$ &$+$\\
    &(0.098) &(0.145) &(0.1)   &(0.094) &(0.115)\\
\end{tabular}
\caption{Estimated right coefficient matrix $\hB$ of MAR(1) using LS method.
Standard errors are shown in parentheses. The right panel indicates
the positively significant, negatively significant and
insignificant  parameters at 5\% level using symbols $(+,-,0)$, respectively.}
\label{table:RR}
\end{table}

Table~\ref{table:RR} shows the estimated $\hB$. Its effect
should be considered in the view of $\hB\hX_t'$. It is seen that
the influence of US's last quarter's indicators
on the current quarter's indicators
of all countries (shown by the first column in $\hat{\hB}$)
are very significantly positive and all larger than those of all other
countries. This is intuitively correct as US is the world's largest
economy. Although it is understandable that Canada has a relatively
small influence on other countries (shown by the last column),
it is surprising to see that Germany
has almost no influence (shown by the second column).
Most of the large coefficients are positive,
showing positive influences among the countries. On the other hand, UK 
has a similar influence pattern as the US (the fourth column), a feature 
that is intuitively difficult to explain. 

Figures~\ref{fig:11}, \ref{fig:12} and \ref{fig:13} show the shock-first impulse
response functions with orthogonal innovations (s1-oIRF) with one standard
deviation shock on US interest rate, US GDP and US CPI,
respectively, using that given in Section~2.2. 
The dotted horizontal lines mark the
values $(0.1,0,-0.1)$ and the dotted vertical lines marks the time
$(0,2,4,6,8,10)$.  
It can be seen from Figure~\ref{fig:11} that a US interest rate shock would be responded positively by the interest rates in other countries in similar patterns, and the impact lasts about a year. 
It is interesting to see that GDP of all countries responds positively to the interest rate shocks at first, and then negatively after two quarters, though very slightly.  
CPI does not have much response to interest rate shocks, but mostly in the negative direction.

From Figure~\ref{fig:12}, it is seen that the interest rate, GDP growth and Production growth of all countries respond positively to a US GDP shock, whose impacts last about 10 quarters. Again, CPI almost does not respond.

On the other hand, Figure~\ref{fig:13} shows that a shock on US CPI generates strong positive responses from CPI of all other countries, while its impacts on interest rates, GDP growth and Production growth are also positive, but relatively small. These patterns are consistent with our interpretations on the matrix $\hA$, reported in Table~\ref{table:LL}. 

\begin{figure}[h]
  \centering
  \includegraphics[width=13cm]{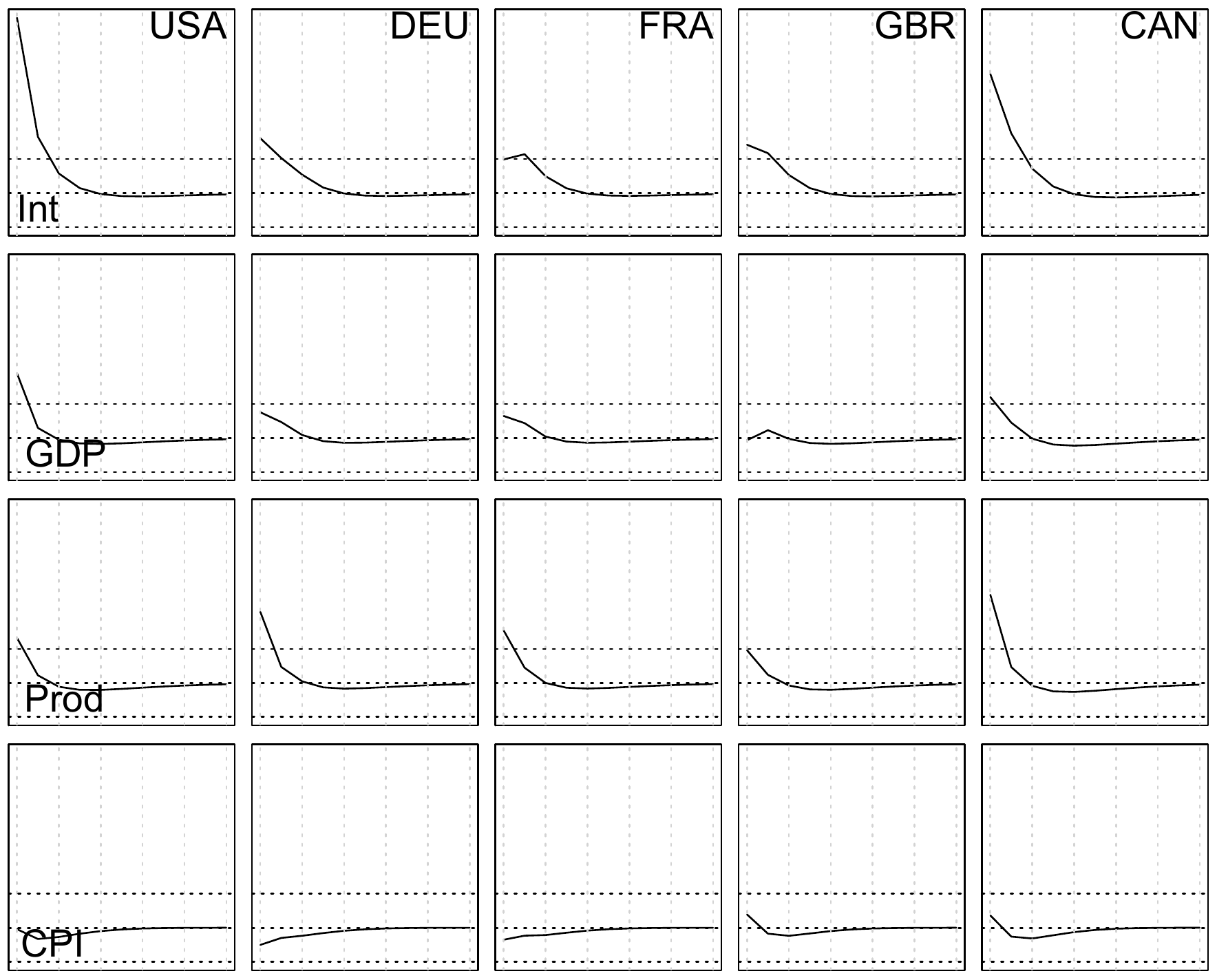}
  \caption{s1-oIRF of MAR(1) model with a unit variance
shock on US interest rate.}
  \label{fig:11}
\end{figure}

\begin{figure}[h]
  \centering
  \includegraphics[width=13cm]{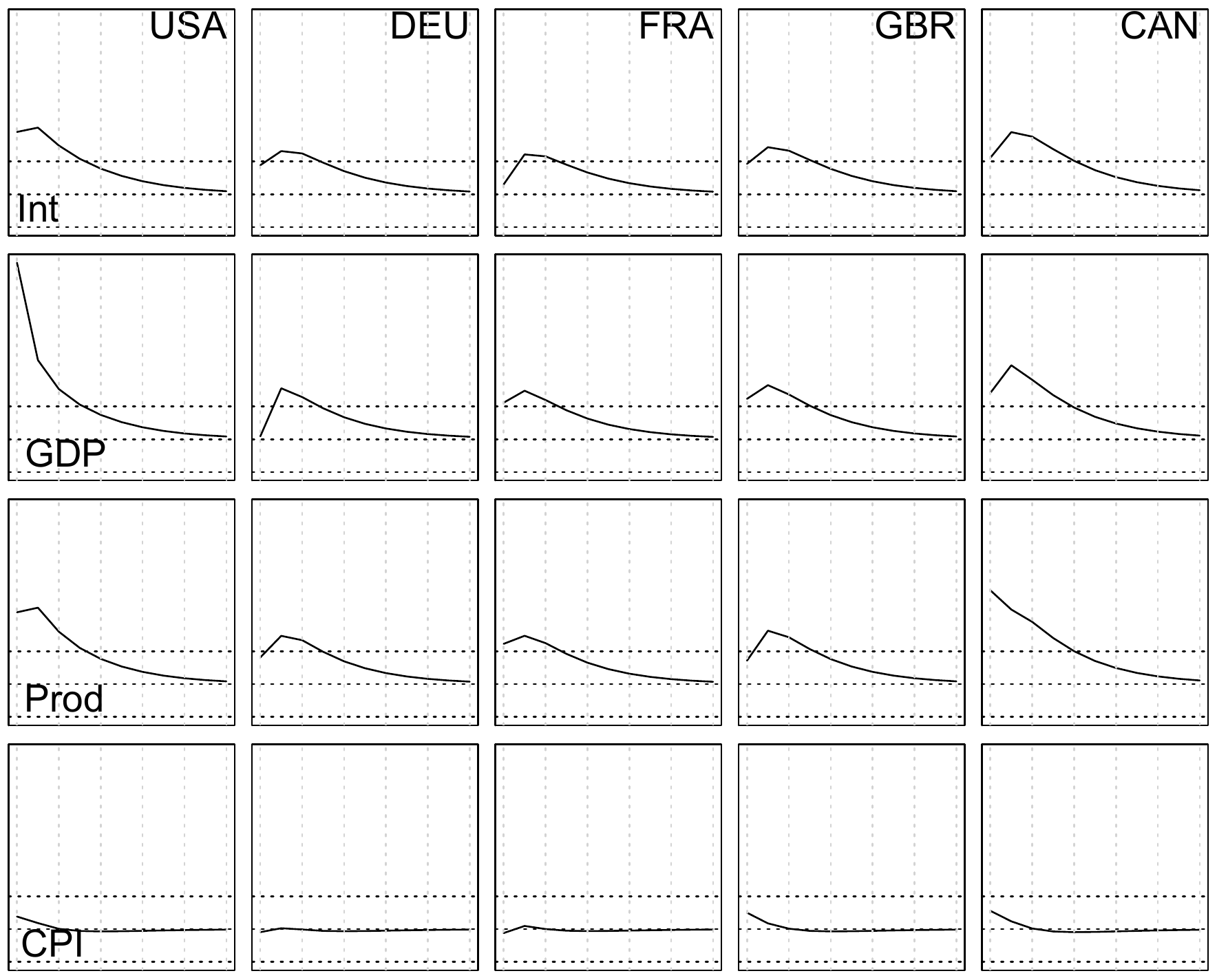}
  \caption{s1-oIRF of MAR(1) model with a unit variance
shock on US GDP rate.}
  \label{fig:12}
\end{figure}

\begin{figure}[h]
  \centering
  \includegraphics[width=13cm]{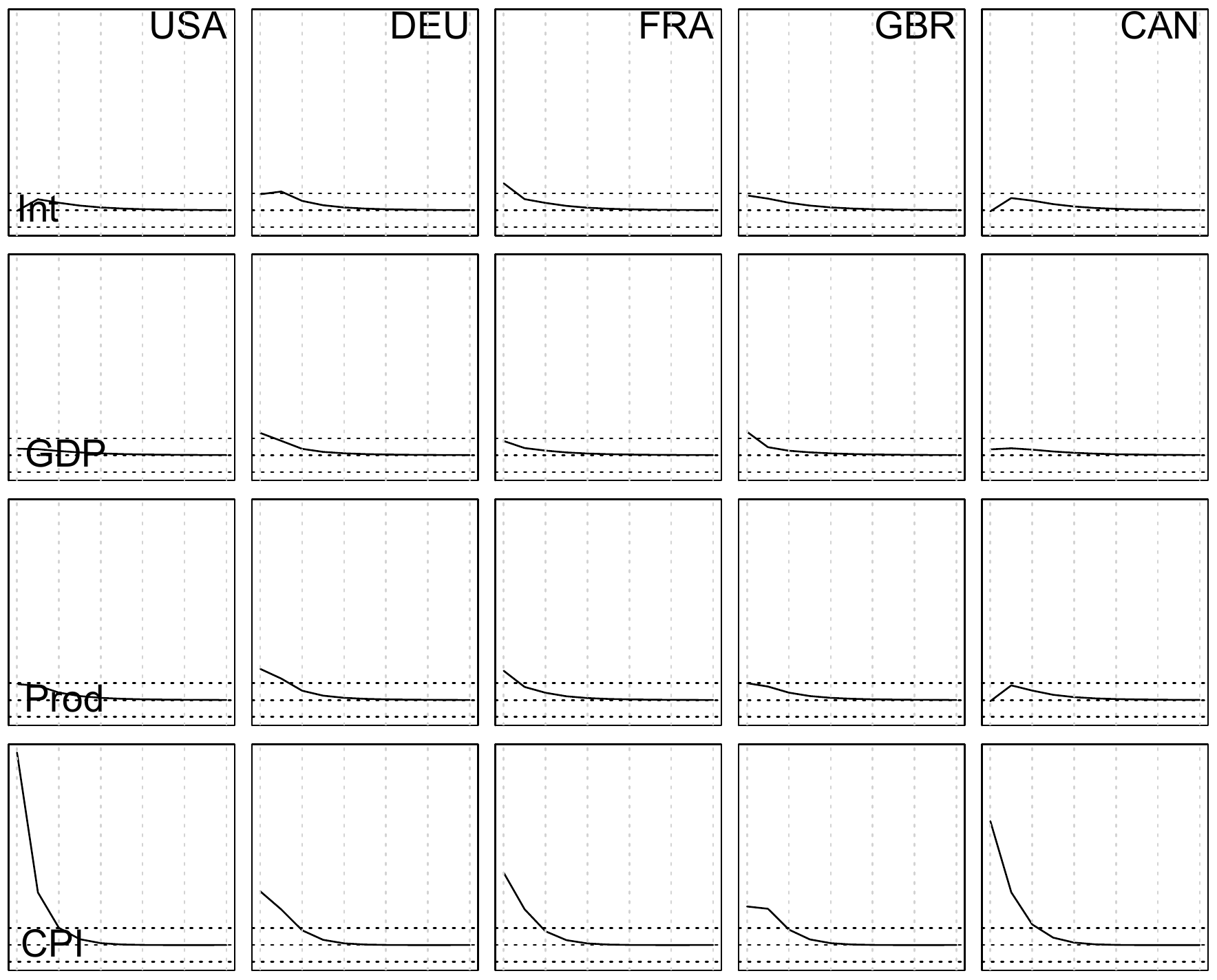}
  \caption{s1-oIRF function of MAR(1) model with a unit variance
shock on US CPI rate.}
  \label{fig:13}
\end{figure}

Figure~\ref{fig:19} shows the residual plots of the MAR(1) estimated
using LS method. There are some outliers.  
Note that the analysis was done by scaling each indicator of all countries (each row) to unit sample variance. Hence the scale of the residuals (Figure 9) are different from the original data plot (Figure 1). As an illustration of the
MAR(1) model, in this analysis we do not try to do any adjustment.  In
Figure~\ref{fig:20} we plot the autocorrelation function (ACF) of the
20 residual series, after fitting the MAR(1) model using the least
squares method. Figure~\ref{fig:21} shows the ACF plots of the 20
original series. It is seen that the MAR(1) model is able to capture
the serial correlations in the 20 time series simultaneously, and lead
to relatively clean ACF plots of the residuals.
Further model checking excises such as the standard portmanteau test may also be applied to assess the adequacy of the model, 
though more investigation needs to be done for its properties for high dimensional cases such as the model used. Note that 
this example is mainly for demonstration. A more thorough analysis of the data may require a model with more 
Kronecker production terms as in (\ref{multi-term}), or with higher AR orders. 

\begin{figure}[h]
  \centering
  \includegraphics[width=13cm]{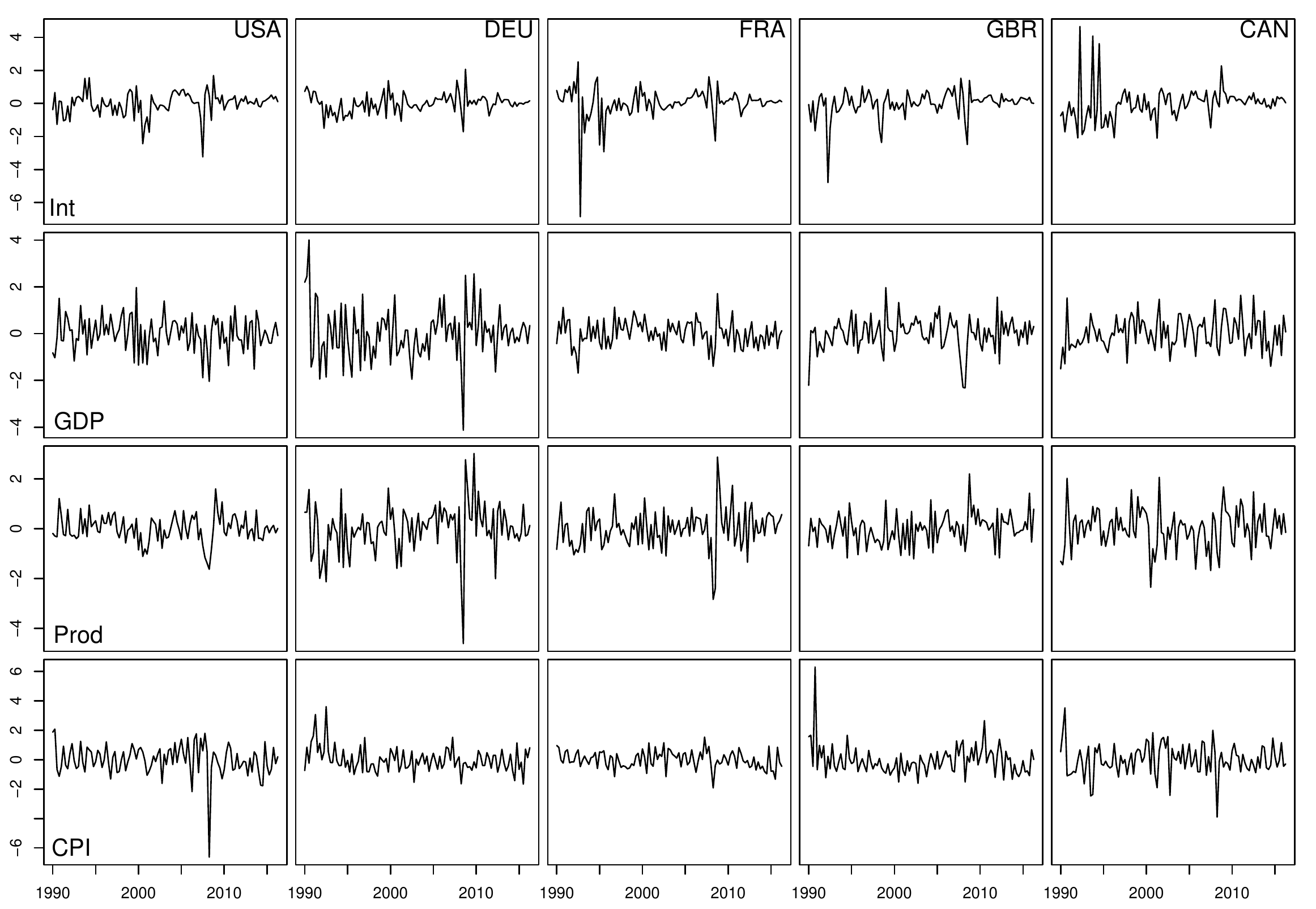}
  \caption{Residual plot of the MAR(1) model.}
  \label{fig:19}
\end{figure}

\begin{figure}[h]
  \centering
  \includegraphics[width=13cm]{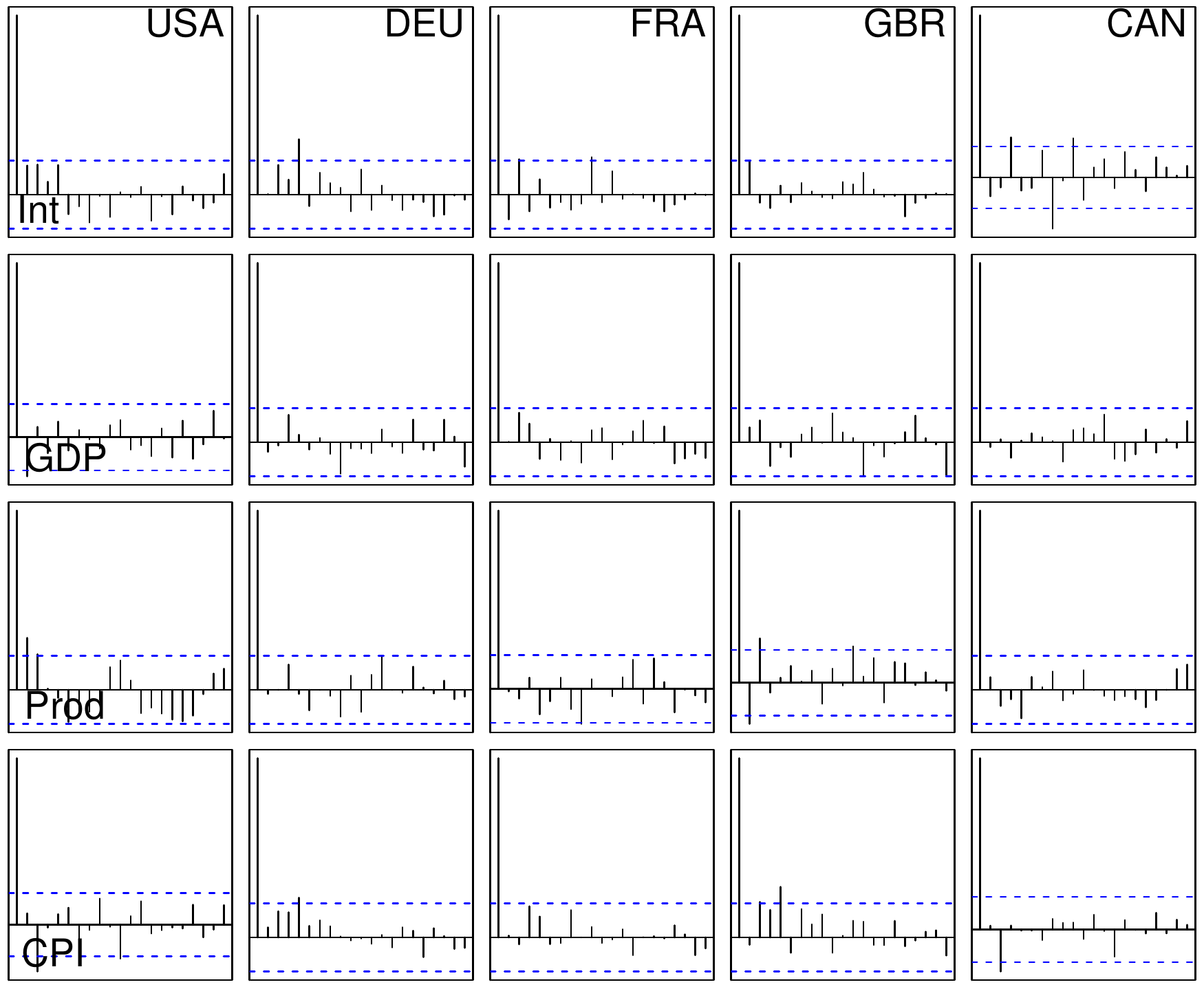}
  \caption{ACF of residuals after fitting MAR(1) model using least squares
method.}
  \label{fig:20}
\end{figure}

\begin{figure}[h]
  \centering
  \includegraphics[width=13cm]{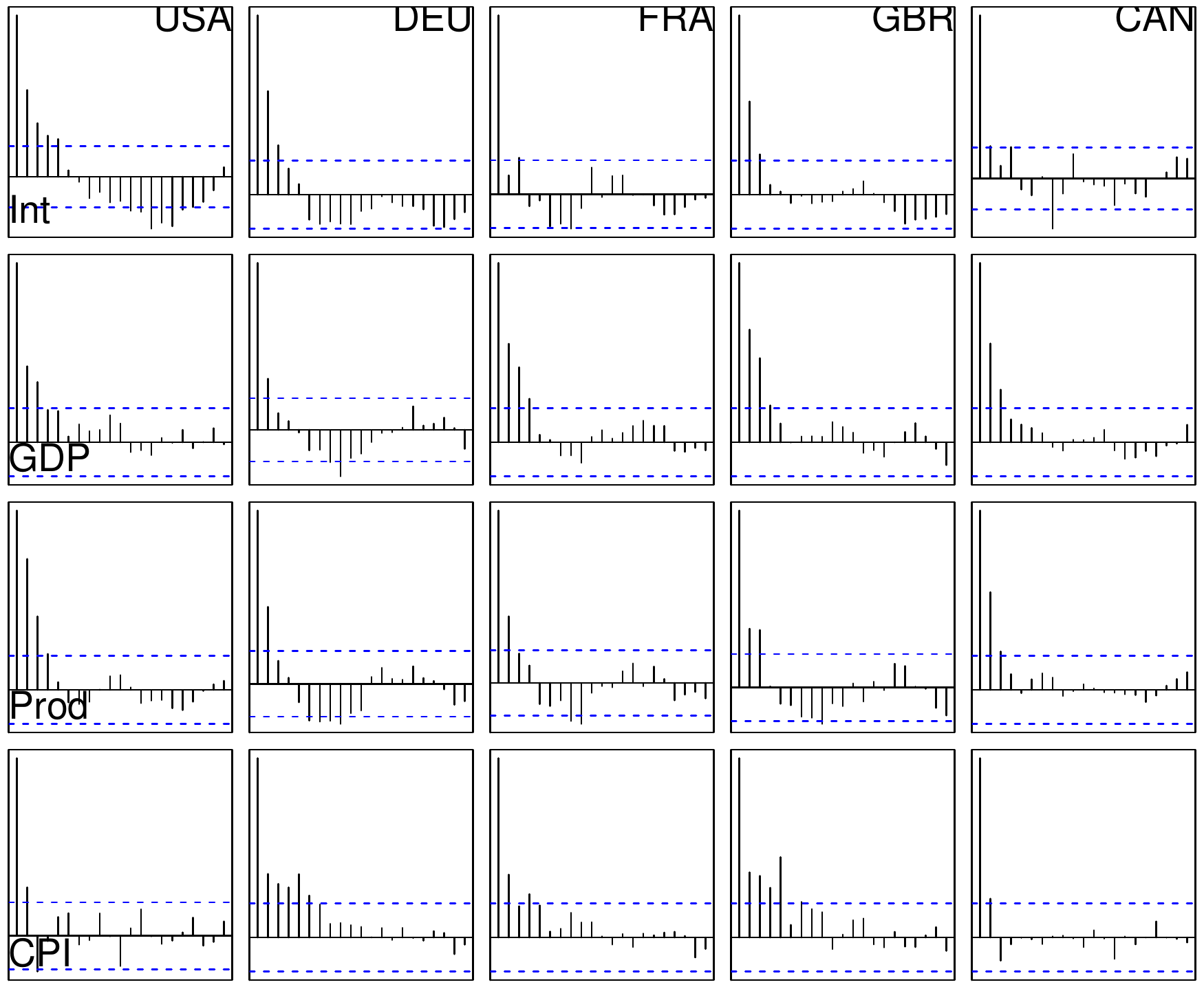}
  \caption{ACF of original series.}
  \label{fig:21}
\end{figure}

We also obtain out-sample rolling forecast performances of the MAR(1)
model as well as univariate AR(1) and AR(2) models for comparison.
Specifically, starting from the last quarter of 2011 ($t=87$) to the
end of the series (the last quarter of 2016, $t=107$), we fit the
corresponding models using all available data at time $t-1$ and
obtained the one step ahead prediction $\hat{\hX}_{t-1}(1)$ for
$\hX_t$ at time $t$.  Sum of prediction error squares
$||\hat{\hX}_{t-1}(1)-\hX_t||_F^2$ of all methods are shown in
Table~\ref{table:pred_SS}. 
It seems that MAR(1) with least squares estimation performs
about 3\% worse than the individual AR(1) models. It is commonly observed in
multivariate time series that the joint model often performs worse
than individual AR models in prediction.  Figure~\ref{indicator_pred}
shows the difference between the sum of squares of prediction error
(or all countries and all indicators) for each quarter between the
MAR(1) model and the individual AR(1) model. It is seen that although
MAR(1) model performs quite poorly in three out of the 20 quarters, it
performs better in the later three years.

Table~\ref{table:pred_SS} also shows that the stacked VAR(1) model performed
terribly in prediction, due to overfitting.

\begin{table}[h] 
\centering
\begin{tabular}{ccc|cc|c}
  MAR(1) PROJ & MAR(1) LSE & MAR(1) MLEs & iAR(1) & iAR(2) & VAR(1) \\ \hline
 146.89 & 141.82 & 136.69 & 136.00 & 135.72 & 296.62\\
\end{tabular}
\caption{Sum of out-of-sample prediction error
squares of MAR(1) model using three different
estimators and the stacked VAR(1) estimator,
and the total sum of out-of-sample prediction error squares of fitting
univariate AR(1) and AR(2) to each individual time series.}
\label{table:pred_SS}
\end{table}

\begin{figure}[h]
  \centering
  \includegraphics[width=13cm,height=11cm]{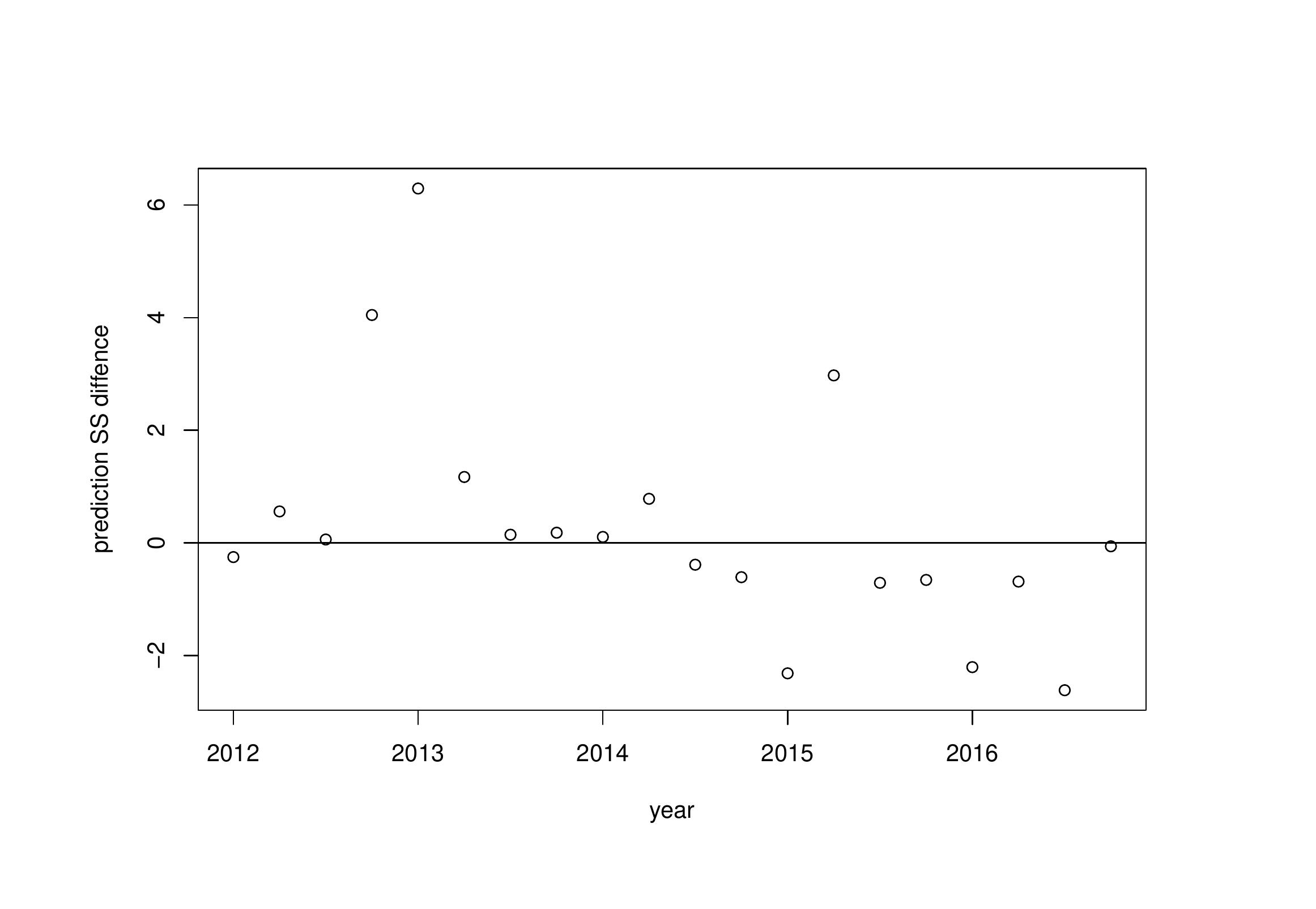}
  \caption{The difference of the sum of prediction error squares between
the MAR(1) model and individual AR(1) model at each quarter.
}
  \label{indicator_pred}
\end{figure}

\section{Conclusion}
\label{sec:con}

We proposed an autoregressive model for matrix-valued time
series in a bilinear form.
It respects the original matrix structure, and provides a much
more parsimonious model, comparing with the direct VAR
approach by stacking the matrix into a long vector. 
Several interpretations of the model, along
with possible extensions are discussed.
Different estimation methods are studied under different
covariance structures of the error matrix. Asymptotic distributions of
the estimators are established, which facilitate the
statistical inferences.

On the other hand, when the matrix observation has large dimensions
itself, our model still involves a large number of parameters,
although much less than that of the corresponding stacked VAR model.
Note that it is natural to have relatively large total
number of parameters, due to the large number of time series involved.
For example, fitting univariate AR(2) models to the $(mn)$ time series
individually  would
require total $2mn$ AR coefficients, while MAR(1) involves $m^2+n^2-1$ AR
coefficients. When $m\sim n$, they use about the same number of parameters.
Also, the structured covariance (\ref{eq:ocov}) involves $m(m+1)/2+n(n+1)/2-1$
parameters while individual AR models involve $mn$ variance parameters,
without considering any correlation among the series. Of course,
regularized estimation approach can be used for MAR(1) model, potentially
shrinking some of the insignificant parameters in the coefficient
matrices to zero, as we have done in a rather {\it ad hoc} way in
Tables~\ref{table:LL} and \ref{table:RR} in the real example.

The impact of dimension $m$ and $n$ on the accuracy of the estimated parameters 
are hidden in the asymptotic variances of the estimators. Of course the larger 
the dimension, the larger the sample size $T$ is required to obtain accurate 
estimates. For very large dimensional matrix time series, \cite{Wang+2018} proposed
a factor model in a bilinear form. The MAR(1) can be used to model the
factor matrix in that of \cite{Wang+2018} to build a dynamic factor model
in matrix form.

There are a number of directions to extend the scope of the proposed model. Sparsity, group sparsity or other structures might be imposed on $\hA$ and $\hB$ to reach a further dimension reduction, so that the model is better suited when both $\hA$ and $\hB$ are of large dimensions. We will also consider MAR models of order larger than one in the future. Furthermore, the idea of MAR can be applied for volatility modeling \citep{engle:1982,bollerslev:1986} as well.

\vspace{0.2in}
\bibliographystyle{apalike}
\bibliography{mybib}

\def\polhk#1{\setbox0=\hbox{#1}{\ooalign{\hidewidth
  \lower1.5ex\hbox{`}\hidewidth\crcr\unhbox0}}}
\begin{thebibliography}{}

\bibitem[Anderson, 1951]{anderson:1951}
Anderson, T.~W. (1951).
\newblock Estimating linear restrictions on regression coefficients for
  multivariate normal distributions.
\newblock {\em Ann. Math. Statistics}, 22:327--351.

\bibitem[Bai and Ng, 2002]{bai:2002}
Bai, J. and Ng, S. (2002).
\newblock Determining the number of factors in approximate factor models.
\newblock {\em Econometrica}, 70(1):191--221.

\bibitem[Basu and Michailidis, 2015]{basu:2015}
Basu, S. and Michailidis, G. (2015).
\newblock Regularized estimation in sparse high-dimensional time series models.
\newblock {\em Ann. Statist.}, 43(4):1535--1567.

\bibitem[Bollerslev, 1986]{bollerslev:1986}
Bollerslev, T. (1986).
\newblock Generalized autoregressive conditional heteroskedasticity.
\newblock {\em J. Econometrics}, 31(3):307--327.

\bibitem[Brockwell and Davis, 1991]{brockwell:1991}
Brockwell, P.~J. and Davis, R.~A. (1991).
\newblock {\em Time series: theory and methods}.
\newblock Springer Series in Statistics. Springer-Verlag, New York, second
  edition.

\bibitem[Cox et~al., 2015]{cox:2015}
Cox, D.~A., Little, J., and O'Shea, D. (2015).
\newblock {\em Ideals, varieties, and algorithms}.
\newblock Undergraduate Texts in Mathematics. Springer, Cham, fourth edition.
\newblock An introduction to computational algebraic geometry and commutative
  algebra.

\bibitem[Davis and Song, 2012]{davis2012noncausal}
Davis, R. and Song, L. (2012).
\newblock Noncausal vector ar processes with application to economic time
  series.
\newblock {\em DP Columbia University}.

\bibitem[{Davis} et~al., 2012]{davis:2012}
{Davis}, R.~A., {Zang}, P., and {Zheng}, T. (2012).
\newblock {Sparse Vector Autoregressive Modeling}.
\newblock {\em ArXiv e-prints}.

\bibitem[Deistler et~al., 1978]{deistler:1978}
Deistler, M., Dunsmuir, W., and Hannan, E.~J. (1978).
\newblock Vector linear time series models: corrections and extensions.
\newblock {\em Adv. in Appl. Probab.}, 10(2):360--372.

\bibitem[Diebold et~al., 2008]{diebold2008global}
Diebold, F.~X., Li, C., and Yue, V.~Z. (2008).
\newblock Global yield curve dynamics and interactions: a dynamic
  nelson--siegel approach.
\newblock {\em Journal of Econometrics}, 146(2):351--363.

\bibitem[Dunsmuir and Hannan, 1976]{dunsmuir:1976}
Dunsmuir, W. and Hannan, E.~J. (1976).
\newblock Vector linear time series models.
\newblock {\em Advances in Appl. Probability}, 8(2):339--364.

\bibitem[Engle, 1982]{engle:1982}
Engle, R.~F. (1982).
\newblock Autoregressive conditional heteroscedasticity with estimates of the
  variance of {U}nited {K}ingdom inflation.
\newblock {\em Econometrica}, 50(4):987--1007.

\bibitem[Fan et~al., 2013]{fan:2013}
Fan, J., Liao, Y., and Mincheva, M. (2013).
\newblock Large covariance estimation by thresholding principal orthogonal
  complements.
\newblock {\em J. R. Stat. Soc. Ser. B. Stat. Methodol.}, 75(4):603--680.

\bibitem[Forni et~al., 2005]{forni:2005}
Forni, M., Hallin, M., Lippi, M., and Reichlin, L. (2005).
\newblock The generalized dynamic factor model: one-sided estimation and
  forecasting.
\newblock {\em J. Amer. Statist. Assoc.}, 100(471):830--840.

\bibitem[Giannone et~al., 2008]{giannone2008nowcasting}
Giannone, D., Reichlin, L., and Small, D. (2008).
\newblock Nowcasting: The real-time informational content of macroeconomic
  data.
\newblock {\em Journal of Monetary Economics}, 55(4):665--676.

\bibitem[Golub and Pereyra, 1973]{golub:1973}
Golub, G.~H. and Pereyra, V. (1973).
\newblock The differentiation of pseudo-inverses and nonlinear least squares
  problems whose variables separate.
\newblock {\em SIAM Journal on Numerical Analysis}, 10(2):413--432.

\bibitem[{Guo} et~al., 2015]{guo:2015}
{Guo}, S., {Wang}, Y., and {Yao}, Q. (2015).
\newblock {High Dimensional and Banded Vector Autoregressions}.
\newblock {\em ArXiv e-prints}.

\bibitem[Hall and Heyde, 1980]{hall:1980}
Hall, P. and Heyde, C.~C. (1980).
\newblock {\em Martingale limit theory and its application}.
\newblock Academic Press Inc. [Harcourt Brace Jovanovich Publishers], New York.
\newblock Probability and Mathematical Statistics.

\bibitem[Hallin and Li{\v{s}}ka, 2011]{hallin2011dynamic}
Hallin, M. and Li{\v{s}}ka, R. (2011).
\newblock Dynamic factors in the presence of blocks.
\newblock {\em Journal of Econometrics}, 163(1):29--41.

\bibitem[Han et~al., 2015]{han:2015}
Han, F., Lu, H., and Liu, H. (2015).
\newblock A direct estimation of high dimensional stationary vector
  autoregressions.
\newblock {\em J. Mach. Learn. Res.}, 16:3115--3150.

\bibitem[Han et~al., 2016]{han:2016}
Han, F., Xu, S., and Liu, H. (2016).
\newblock Rate-optimal estimation of high dimensional time series.
\newblock Technical report, Washington University, Department of Statistics.

\bibitem[Hannan, 1970]{hannan:1970}
Hannan, E.~J. (1970).
\newblock {\em Multiple time series}.
\newblock John Wiley and Sons, Inc., New York-London-Sydney.

\bibitem[Horn and Johnson, 1994]{horn:1994}
Horn, R.~A. and Johnson, C.~R. (1994).
\newblock {\em Topics in matrix analysis}.
\newblock Cambridge University Press, Cambridge.
\newblock Corrected reprint of the 1991 original.

\bibitem[Horn and Johnson, 2012]{horn2012matrix}
Horn, R.~A. and Johnson, C.~R. (2012).
\newblock {\em Matrix analysis}.
\newblock Cambridge university press.

\bibitem[Hosking, 1980]{hosking:1980}
Hosking, J. R.~M. (1980).
\newblock The multivariate portmanteau statistic.
\newblock {\em J. Amer. Statist. Assoc.}, 75(371):602--608.

\bibitem[Hosking, 1981a]{hosking:1981}
Hosking, J. R.~M. (1981a).
\newblock Equivalent forms of the multivariate portmanteau statistic.
\newblock {\em J. Roy. Statist. Soc. Ser. B}, 43(2):261--262.

\bibitem[Hosking, 1981b]{hosking:1981lm}
Hosking, J. R.~M. (1981b).
\newblock Lagrange-multiplier tests of multivariate time-series models.
\newblock {\em J. Roy. Statist. Soc. Ser. B}, 43(2):219--230.

\bibitem[Izenman, 1975]{izenman:1975}
Izenman, A.~J. (1975).
\newblock Reduced-rank regression for the multivariate linear model.
\newblock {\em J. Multivariate Anal.}, 5:248--264.

\bibitem[Kaufman, 1975]{kaufman:1975}
Kaufman, L. (1975).
\newblock A variable projection method for solving separable nonlinear least
  squares problems.
\newblock {\em BIT Numerical Mathematics}, 15(1):49--57.

\bibitem[Kock and Callot, 2015]{kock:2015}
Kock, A.~B. and Callot, L. (2015).
\newblock Oracle inequalities for high dimensional vector autoregressions.
\newblock {\em J. Econometrics}, 186(2):325--344.

\bibitem[Lam and Yao, 2012]{lam:2012}
Lam, C. and Yao, Q. (2012).
\newblock Factor modeling for high-dimensional time series: inference for the
  number of factors.
\newblock {\em Ann. Statist.}, 40(2):694--726.

\bibitem[Lam et~al., 2011]{lam:2011}
Lam, C., Yao, Q., and Bathia, N. (2011).
\newblock Estimation of latent factors for high-dimensional time series.
\newblock {\em Biometrika}, 98(4):901--918.

\bibitem[Lanne and Saikkonen, 2013]{lanne2013noncausal}
Lanne, M. and Saikkonen, P. (2013).
\newblock Noncausal vector autoregression.
\newblock {\em Econometric Theory}, 29(3):447--481.

\bibitem[Li and McLeod, 1981]{li:1981}
Li, W.~K. and McLeod, A.~I. (1981).
\newblock Distribution of the residual autocorrelations in multivariate {ARMA}
  time series models.
\newblock {\em J. Roy. Statist. Soc. Ser. B}, 43(2):231--239.

\bibitem[L{\"u}tkepohl, 2005]{lutkepohl:2005}
L{\"u}tkepohl, H. (2005).
\newblock {\em New introduction to multiple time series analysis}.
\newblock Springer-Verlag, Berlin.

\bibitem[Moench et~al., 2013]{moench2013dynamic}
Moench, E., Ng, S., and Potter, S. (2013).
\newblock Dynamic hierarchical factor models.
\newblock {\em Review of Economics and Statistics}, 95(5):1811--1817.

\bibitem[Nardi and Rinaldo, 2011]{nardi:2011}
Nardi, Y. and Rinaldo, A. (2011).
\newblock Autoregressive process modeling via the {L}asso procedure.
\newblock {\em J. Multivariate Anal.}, 102(3):528--549.

\bibitem[Negahban and Wainwright, 2011]{negahban:2011}
Negahban, S. and Wainwright, M.~J. (2011).
\newblock Estimation of (near) low-rank matrices with noise and
  high-dimensional scaling.
\newblock {\em Ann. Statist.}, 39(2):1069--1097.

\bibitem[{Nicholson} et~al., 2015]{nicholson:2015}
{Nicholson}, W., {Matteson}, D., and {Bien}, J. (2015).
\newblock {VARX-L: Structured Regularization for Large Vector Autoregressions
  with Exogenous Variables}.
\newblock {\em ArXiv e-prints}.

\bibitem[Poskitt and Tremayne, 1982]{poskitt:1982}
Poskitt, D.~S. and Tremayne, A.~R. (1982).
\newblock Diagnostic test for multiple time series models.
\newblock {\em Ann. Statist.}, 10(1):114--120.

\bibitem[Raskutti et~al., 2015]{Raskutti+2015}
Raskutti, G., Yuan, M., and Chen, H. (2015).
\newblock {Convex Regularization for High-Dimensional Multi-Response Tensor
  Regression}.
\newblock Technical report.

\bibitem[{Song} and {Bickel}, 2011]{song:2011b}
{Song}, S. and {Bickel}, P.~J. (2011).
\newblock {Large Vector Auto Regressions}.
\newblock {\em ArXiv e-prints}.

\bibitem[Tiao and Box, 1981]{tiao:1981}
Tiao, G.~C. and Box, G. E.~P. (1981).
\newblock Modeling multiple time series with applications.
\newblock {\em J. Amer. Statist. Assoc.}, 76(376):802--816.

\bibitem[Tsai and Tsay, 2010]{tsai:2010}
Tsai, H. and Tsay, R.~S. (2010).
\newblock Constrained factor models.
\newblock {\em J. Amer. Statist. Assoc.}, 105(492):1593--1605.
\newblock Supplementary materials available online.

\bibitem[Tsay, 2014]{tsay:2014}
Tsay, R.~S. (2014).
\newblock {\em Multivariate time series analysis}.
\newblock Wiley Series in Probability and Statistics. John Wiley \& Sons, Inc.,
  Hoboken, NJ.

\bibitem[Van~Loan, 2000]{vanloan:2000}
Van~Loan, C. (2000).
\newblock The ubiquitous kronecker product.
\newblock {\em Journal of Computational and Applied Mathematics}, 123(1):85 --
  100.
\newblock Numerical Analysis 2000. Vol. III: Linear Algebra.

\bibitem[Van~Loan and Pitsianis, 1993]{vanloan:1993}
Van~Loan, C. and Pitsianis, N. (1993).
\newblock Approximation with kronecker products.
\newblock In Moonen, M. and Golub, G., editors, {\em Linear Algebra for Large
  Scale and Real Time Applications}, pages 293--314. Kluwer Publications,
  Dordrecht.

\bibitem[Wang et~al., 2019]{wang:2018}
Wang, D., Liu, X., and Chen, R. (2019).
\newblock Factor models for matrix-valued high-dimensional time series.
\newblock {\em Journal of Econometrics}, 208(1):231 -- 248.

\bibitem[Wang et~al., 2018]{Wang+2018}
Wang, D., Yang, D., Shen, H., and Zhu, H. (2018).
\newblock {On scalar-on-matrix bilinear regression analysis}.
\newblock Technical report.

\bibitem[Zhao and Leng, 2014]{Zhao+2014a}
Zhao, J. and Leng, C. (2014).
\newblock {Structured lasso for regression with matrix covariates}.
\newblock {\em Statist. Sinica}, 24:799--814.

\bibitem[Zhou et~al., 2013]{Zhou+2013a}
Zhou, H., Li, L., and Zhu, H. (2013).
\newblock {Tensor Regression with Applications in Neuroimaging Data Analysis}.
\newblock {\em J. Amer. Statist. Assoc.}, 108(502):540--552.

\end{thebibliography}

\newpage

\noindent
{\bf \LARGE Appendix: Proofs of the Theorems}

\vspace{0.2in}

We collect the proofs of Proposition~\ref{thm:causal},
Theorem~\ref{thm:clt1}, Theorem~\ref{thm:lse}, Theorem~\ref{thm:mle}, Corollary~\ref{thm:eff} and Corollary~\ref{thm:test} in this
section. 

\newcommand{\hG}{{\h{G}}}

\setcounter{subsection}{0}
\def\thesubsection{A.\arabic{subsection}}

\subsection{Basics}

We begin by listing some basic properties of the Kronecker product and
its relationship with linear matrix equations. Let $M_{m,n}$ be the
set of all $m\times n$ matrices over the field of complex numbers
$\C$. The Kronecker product of $\hC=(c_{ij})\in M_{m,n}$, and
$\hD=(d_{ij})\in M_{p,q}$, denoted by $\hC\otimes\hD$, is defined to be
the block matrix
\begin{equation*}
  \hC\otimes\hD =
  \begin{pmatrix}
    c_{11}\hD & \cdots & c_{1n}\hD \\
    \vdots & \ldots & \vdots \\
    c_{m1}\hD & \cdots & c_{mn}\hD
  \end{pmatrix}\in M_{mp,nq}.
\end{equation*}
In the following proposition, we list some facts regarding the
Kronecker product, which are used in this section at various places
without specific references. Proofs of these facts can be found in
Chapter~4 of \cite{horn:1994}.
\begin{proposition}
  \label{thm:kronecker}
  Let $\hC\in M_{m,n}$, $\hD\in M_{p,q}$, $\hF\in M_{n,k}$,
  $\hG \in M_{q,l}$ and $\hZ\in M_{n,p}$.
  \begin{enumerate}
    \renewcommand{\labelenumi}{(\roman{enumi})}
  \item $(\hC\otimes\hD)'=\hC'\otimes\hD'$.
  \item If both $\hC$ and $\hD$ are invertible square matrices, then
    $\hC\otimes\hD$ is also invertible, and
    $(\hC\otimes\hD)^{-1}=\hC^{-1}\otimes\hD^{-1}$.
  \item $(\hC\otimes\hD)(\hF\otimes\hG)=(\hC\hF)\otimes(\hD\hG)$.
  \item $\vect(\hC\hZ\hD)=(\hD'\otimes\hC)\vect(\hZ)$.
  \item $\mathrm{rank}(\hC\otimes\hD)=\mathrm{rank}(\hD\otimes\hC)=\mathrm{rank}(\hC)\cdot\mathrm{rank}(\hD)$.
  \item Let $\hC\in M_{m,m}$ and $\hD\in M_{n,n}$. Let
    $\{\lambda_1,\lambda_2,\ldots,\lambda_m\}$ be eigenvalues of $\hC$
    (including multiplicities), and $\{\eta_1,\eta_2,\ldots,\eta_n\}$
    be eigenvalues of $\hD$. The $mn$ eigenvalues (including
    multiplicities) of $\hC\otimes\hD$ are
    $\{\lambda_i\eta_j:\;1\leq i\leq m,\,1\leq j\leq n\}$.
  \end{enumerate}
\end{proposition}

\begin{proof}[Proof of Proposition~\ref{thm:causal}]
  It is known that the VAR(1) model in \eqref{eq:var1} admits a
  stationary and causal solution if the spectral radius of the
  coefficient matrix $\Phi$ is strictly less than 1 \citep[See for
  example, \S11.3 of ][]{brockwell:1991}.  By
  Proposition~\ref{thm:kronecker}$\sim$(vi), all the eigenvalues of
  $\hB\otimes\hA$ are of the form $\lambda_i\eta_j$, where $\lambda_i$
  and $\eta_j$ are the eigenvalues of $\hA$ and $\hB$ respectively. As
  a consequence, $\rho(\hB\otimes\hA)=\rho(\hA)\cdot\rho(\hB)$. Since
  the MAR(1) model in \eqref{eq:mar1} can be represented as a VAR
  model as given by \eqref{eq:ovar1}, the proposition then follows.
\end{proof}

\subsection{Proof of Theorem~\ref{thm:clt1}}

\newcommand{\hubeta}{\underline{\hbeta}}

\begin{proof}[Proof of Theorem~\ref{thm:clt1}]
  Let $\halpha=\vect(\hA)$ and $\hubeta=\vect(\hB)$. Note that the
  convention that $\|\hA\|_F=1$ is equivalent to $\|\halpha\|=1$. Also note that since $\hbeta$ is used to denote $\vect(\hB')$ in Theorem~\ref{thm:lse}, we use $\hubeta$ here for $\vect(\hB)$. Recall that $\hbeta_1$ is the normalized version of $\hubeta$. The gradient condition of the NKP problem \eqref{eq:nkp} is given by
  \begin{equation}
    \label{eq:nkp_grad}
    \begin{aligned}
      \hat\halpha\hat\hubeta'\hat\hubeta-\tilde\Phi\hat\hubeta&=0 \\
      \hat\hubeta\hat\halpha'\hat\halpha-\tilde\Phi'\hat\halpha&=0.
    \end{aligned}
  \end{equation}
  Recall that we require $\|\hat\hA_1\|_F=1$, so we similarly also
  require that the solution of the NKP problem satisfies
  $\|\hat\halpha\|=1$. Since both $\|\halpha\|=1$ and
  $\|\hat\halpha\|=1$ it follows that
  $(\hat\halpha-\halpha)'\halpha=o_P(T^{-1/2})$. Replacing
  $\tilde\Phi$ by $\halpha\hubeta'+(\tilde\Phi-\halpha\hubeta')$ in the
  gradient conditions, we have
  \begin{equation}
    \label{eq:nkp_grad1}
    \begin{aligned}
      (\hat\halpha-\halpha)\hubeta'\hubeta+\halpha(\hat\hubeta-\hubeta)'\hubeta
      & = (\tilde\Phi-\halpha\hubeta')\hubeta + o_P(T^{-1/2})\\
      \hat\hubeta-\hubeta&=(\tilde\Phi-\halpha\hubeta')'\halpha + o_P(T^{-1/2}).
    \end{aligned}
  \end{equation}
  It follows that
  \begin{equation}
    \label{eq:3}
    \begin{pmatrix}
      \vect(\hat\hA_1-\hA) \\
      \vect(\hat\hB_1-\hB)
    \end{pmatrix} =\hV_0\mathrm{vec}(\tilde\Phi-\halpha\hubeta') + o_P(T^{-1/2}),
  \end{equation}
  and
  \begin{equation}
    \label{eq:4}
    \hat\halpha\hat\hubeta' - \halpha\hubeta' =
    (\tilde\Phi-\halpha\hubeta')\hbeta_1\hbeta_1'+\alpha\alpha'(\tilde\Phi-\halpha\hubeta')
    -\alpha\alpha'(\tilde\Phi-\halpha\hubeta')\hbeta_1\hbeta_1' + o_P(T^{-1/2}).
  \end{equation}
  The first central limit theorem stated in Theorem~\ref{thm:clt1} is
  an immediate consequence of \eqref{eq:3}. Taking vectorization on
  both sides of \eqref{eq:4}, we have
  \begin{equation*}
    \hat\hubeta\otimes\hat\halpha - \hubeta\otimes\halpha
    = \hV_1\mathrm{vec}(\tilde\Phi-\halpha\hubeta') + o_P(T^{-1/2}),
  \end{equation*}
  and the second central limit theorem follows.
\end{proof}

\subsection{Proof of Theorem~\ref{thm:lse}}

To prove Theorem~\ref{thm:lse}, we first state and prove the
following lemma.
\begin{lemma}
  \label{thm:lemma1}
  Consider the VAR(1) representation of \eqref{eq:ovar1}, and let
  $\Phi=\hB\otimes\hA$. Assume the conditions of
  Theorem~\ref{thm:lse}. Then for any sequence $\{c_T\}$ such that $c_T\rightarrow\infty$,
  \begin{equation}
    \label{eq:diff}
    P\left[\inf_{\sqrt{T}\|\bar\Phi-\Phi\|_F\geq c_T}\sum_{t=2}^T
      \left\|\vect(\hX_t)-\bar\Phi\vect(\hX_{t-1})\right\|^2
      \leq\sum_{t=2}^T\left\|\vect(\hE_t)\right\|^2\right] \rightarrow 0.
  \end{equation}
\end{lemma}
\begin{proof}
  First of all, by the ergodic theorem
  \begin{equation*}
    \frac{1}{T}\sum_{t=2}^T\vect(\hX_{t-1})\vect(\hX_{t-1})' \rightarrow \Gamma_0\quad\hbox{a.s.}
  \end{equation*}
  It follows that for any constant $c>0$,
  \begin{equation}
    \label{eq:aslim1}
  \begin{aligned}
    \sup_{\sqrt{T}\|\bar\Phi-\Phi\|_F\leq c}
    & \left|\sum_{t=2}^T\tr\left[(\bar\Phi-\Phi)\vect(\hX_{t-1})\vect(\hX_{t-1})'(\bar\Phi-\Phi)'\right]\right.\\
    & \left.\phantom{\sum_{t=2}^T\!\!\!\!}-T\cdot\tr\left[(\bar\Phi-\Phi)\Gamma_0(\bar\Phi-\Phi)'\right]\right| \rightarrow 0\quad\hbox{a.s.}
  \end{aligned}    
  \end{equation}
  In \eqref{eq:aslim1}, the superme is taken over $\bar\Phi$. As a
  consequence of \eqref{eq:aslim1}, there exists a sequence $\{c_T'\}$
  such that $c_T'\rightarrow \infty$, $c_T'\leq c_T$, and
  \begin{equation}
    \label{eq:ip1}
  \begin{aligned}
    \sup_{\sqrt{T}\|\bar\Phi-\Phi\|_F\leq c_T'}
    &\left|\sum_{t=2}^T\tr\left[(\bar\Phi-\Phi)\vect(\hX_{t-1})\vect(\hX_{t-1})'(\bar\Phi-\Phi)'\right]\right.\\
    &\left.\phantom{\sum_{t=2}^T\!\!\!\!}-T\cdot\tr\left[(\bar\Phi-\Phi)\Gamma_0(\bar\Phi-\Phi)'\right]\right| \rightarrow 0\quad\hbox{in probability}.
  \end{aligned}    
  \end{equation}
  Now we write
  \begin{equation}
    \label{eq:diff1}
    \begin{aligned}
      &\sum_{t=2}^T\left\|\vect(\hX_t)-\bar\Phi\vect(\hX_{t-1})\right\|^2 -
      \sum_{t=2}^T\left\|\vect(\hE_t)\right\|^2 \\
      =&-2\sum_{t=2}^T\tr\left[(\bar\Phi-\Phi)\vect(\hX_{t-1})\vect(\hE_t)'\right]
      +\sum_{t=2}^T\tr\left[(\bar\Phi-\Phi)\vect(\hX_{t-1})\vect(\hX_{t-1})'(\bar\Phi-\Phi)'\right].
    \end{aligned}
  \end{equation}
  On the boundary set $\sqrt{T}\|\bar\Phi-\Phi\|_F=c_T'$, by
  calculating the variance, we know that
  \begin{equation}
    \label{eq:martingale}
    \sum_{t=2}^T\tr\left[(\bar\Phi-\Phi)\vect(\hX_{t-1})\vect(\hE_t)'\right] = O_P(c_T').
  \end{equation}
  On the other hand, on the boundary set
  $\sqrt{T}\|\bar\Phi-\Phi\|_F=c_T'$,
  \begin{equation}
    \label{eq:1}
    T\cdot\tr\left[(\bar\Phi-\Phi)\Gamma_0(\bar\Phi-\Phi)'\right]
    \geq{\lambda_{\min}}(\Gamma_0)(c_T')^2,
  \end{equation}
  where ${\lambda_{\min}}(\Gamma_0)$ is the minimum eigenvalue of
  $\Gamma_0$, which is strictly positive under the condition that
  $\hA$, $\hB$ and $\Sigma$ are nonsingular. Combining
  \eqref{eq:ip1}$\sim$\eqref{eq:1}, and with fact that
  $c_T'\rightarrow\infty$, we have
  \begin{equation}
    \label{eq:5}
    P\left[\inf_{\sqrt{T}\|\bar\Phi-\Phi\|_F= c_T'}\sum_{t=2}^T
      \left\|\vect(\hX_t)-\bar\Phi\vect(\hX_{t-1})\right\|^2
      \leq\sum_{t=2}^T\left\|\vect(\hE_t)\right\|^2\right] \rightarrow 0.
  \end{equation}
  Observe that
  $\sum_{t=2}^T\left\|\vect(\hX_t)-\bar\Phi\vect(\hX_{t-1})\right\|^2$
  is a convex function of $\bar\Phi$, so \eqref{eq:diff} is implied by
  \eqref{eq:5} and the convexity.
\end{proof}

Now we are ready to give the proof of Theorem~\ref{thm:lse}.
\begin{proof}[Proof of Theorem~\ref{thm:lse}]
  Let
  $\mathbb{S}=\{\hC:\,\hC \hbox{ is a } m\times m \hbox{ matrix, and }
  \|\hC\|_F=1\}$. Let $\{c_T\}$ be any sequence such that
  $c_T\rightarrow\infty$, and $c_T/\sqrt{T}\rightarrow 0$. By the
  conditions that $\hA$ and $\hB$ are nonsingular, and
  $\hA\in\mathbb{S}$, it can be show that if $\bar\hA$ and $\bar\hB$
  are such that $\bar\hA\in\mathbb{S}$, and
  $T\|\bar\hA-\hA\|_F^2+T\|\bar\hB-\hB\|_F^2\geq c_T^2$, then
  $\|\bar\hB\otimes\bar\hA-\hB\otimes\hA\|_F\geq C\cdot c_T$, where
  $C$ is a constant determined by $\hA$ and $\hB$. By
  Lemma~\ref{thm:lemma1}, we have
  \begin{equation*}
    P\left[\min_{T\|\bar\hA-\hA\|_F^2+T\|\bar\hB-\hB\|_F^2\geq c_T^2}\sum_{t=2}^T
      \left\|\vect(\hX_t)-(\bar\hB\otimes\bar\hA)\vect(\hX_{t-1})\right\|^2
      \leq\sum_{t=2}^T\left\|\vect(\hE_t)\right\|^2\right] \rightarrow 0,
  \end{equation*}
  with the implicit requirement that $\bar\hA\in\mathbb{S}$. It
  follows that
  \begin{equation}
    \label{eq:2}
    P\left[T\|\hat\hA-\hA\|_F^2+T\|\hat\hB-\hB\|_F^2\leq c_T^2\right] \rightarrow 1,
  \end{equation}
  also with the implicit requirement that $\hat\hA\in\mathbb{S}$.
  Since \eqref{eq:2} holds for any sequence $\{c_T\}$ such that
  $c_T\rightarrow\infty$, and $c_T/\sqrt{T}\rightarrow 0$, we have
  \begin{equation*}
    \hat\hA= \hA+O_P(T^{-1/2}),\quad\hbox{and}\quad\hat\hB= \hB+O_P(T^{-1/2}).
  \end{equation*}
  
  We now repeat the gradient condition \eqref{eq:lse_grad} here:
  \begin{equation}
    \label{eq:lse_grad1}
    \begin{aligned}
      \sum_t\hat\hA_2\hX_{t-1}\hat\hB_2'\hat\hB_2\hX_{t-1}'-\sum_t\hX_t\hat\hB_2\hX_{t-1}'&=\hzero \\
      \sum_t\hat\hB_2\hX_{t-1}'\hat\hA_2'\hat\hA_2\hX_{t-1}-\sum_t\hX_t'\hat\hA_2\hX_{t-1}&=\hzero.
    \end{aligned}
  \end{equation}
  Replacing each $\hX_t$ by $\hA\hX_{t-1}\hB'+\hE_t$ in
  \eqref{eq:lse_grad1}, we have
  \begin{equation*}
    \label{eq:lse_grad2}
    \begin{aligned}
      \sum_t(\hat\hA_2-\hA)\hX_{t-1}\hB'\hB\hX_{t-1}'+  \sum_t\hA\hX_{t-1}(\hat\hB_2-\hB)'\hB\hX_{t-1}'
      &=\sum_t\hE_t\hB\hX_{t-1}' + o_P(\sqrt{T})\\
      \sum_t\hX_{t-1}'\hA'(\hat\hA_2-\hA)\hX_{t-1}\hB'+ \sum_t\hX_{t-1}'\hA'\hA\hX_{t-1}(\hat\hB_2-\hB)'
      &=\sum_t\hX_{t-1}'\hA'\hE_t + o_P(\sqrt{T}).
    \end{aligned}
  \end{equation*}
  Taking vectorization on both sides, we have
  \begin{equation*}
    \begin{aligned}
      &\begin{pmatrix}
      \left(\sum_t\hX_{t-1}\hB'\hB\hX_{t-1}'\right)\otimes \hI & \sum_t(\hX_{t-1}\hB')\otimes(\hA\hX_{t-1}) \\
      \sum_t(\hB\hX_{t-1}')\otimes(\hX_{t-1}'\hA') & \hI\otimes\left( \sum_t\hX_{t-1}'\hA'\hA\hX_{t-1} \right)
    \end{pmatrix}
    \begin{pmatrix}
      \vect(\hat\hA_2-\hA) \\
      \vect(\hat\hB_2'-\hB')
    \end{pmatrix} \\
    =&\sum_t
    \begin{pmatrix}
      (\hX_{t-1}\hB')\otimes \hI \\
      \hI\otimes (\hX_{t-1}'\hA')
    \end{pmatrix}
    \vect(\hE_t)+ o_P(\sqrt{T}),
    \end{aligned}
  \end{equation*}
  which can be rewritten as
  \begin{equation}
    \label{eq:lse_grad3}
    \left(\sum_t\hW_{t-1}\hW_{t-1}'\right)
    \begin{pmatrix}
      \vect(\hat\hA_2-\hA) \\
      \vect(\hat\hB_2'-\hB')
    \end{pmatrix}
    =\sum_t\hW_{t-1}\vect(\hE_t)+ o_P(\sqrt{T}).
  \end{equation}
  By the ergodic theorem as $\hX_t$ is strictly stationary with
  i.i.d. innovations under the conditions, we have
  \begin{equation*}
    \frac{1}{T}\sum_t\hW_{t-1}\hW_{t-1}'\rightarrow \E(\hW_t\hW_t'),\quad\hbox{a.s.}
  \end{equation*}
  Observe that $\E(\hW_t\hW_t')$ is not a full rank matrix, because
  $\E(\hW_t\hW_t')(\halpha',-\hbeta')'=\hzero$. On the other hand,
  since we require $\|\hA\|_F=1$ and $\|\hat\hA_2\|_F=1$, it holds
  that $\halpha'(\vect(\hat\hA_2)-\halpha)=o_P(T^{-1/2})$.  and
  consequently
  \begin{equation*}
   \hH \begin{pmatrix}
      \vect(\hat\hA_2-\hA) \\
      \vect(\hat\hB_2'-\hB')
    \end{pmatrix} = \frac{1}{T}\sum_t\hW_{t-1}\vect(\hE_t)+ o_P(T^{-1/2}),
  \end{equation*}
  where $\hH:=\E(\hW_t\hW_t')+\hgamma\hgamma'$.  By martingale central
  limit theorem \citep{hall:1980}
  \begin{equation*}
    \sum_t\hW_{t-1}\vect(\hE_t)\Rightarrow N(\hzero, \E(\hW_t\Sigma\hW_t')).
  \end{equation*}
  It follows that
  \begin{equation*}
    \sqrt{T}\begin{pmatrix}
      \vect(\hat\hA_2-\hA) \\
      \vect(\hat\hB_2'-\hB')
    \end{pmatrix}
    \Rightarrow N(\hzero, \Xi_2),
  \end{equation*}
  where $\Xi_2 := \hH^{-1}\E(\hW_t\Sigma\hW_t')\hH^{-1}$.
  Furthermore, noting that
  \begin{eqnarray*}
 \lefteqn{\vect(\hat{\hB}_2')\otimes\vect(\hat{\hA}_2)
          -\vect(\hB')\otimes\vect(\hA)} \\
 &=& \vect(\hat\hB_2'-\hB')\otimes\halpha + \hbeta\otimes\vect(\hat\hA_2-\hA)\\
 &=& \hV\begin{pmatrix}
      \vect(\hat\hA_2-\hA) \\
      \vect(\hat\hB_2'-\hB')
    \end{pmatrix}
    + o_P(T^{-1/2}),
  \end{eqnarray*}
  where $\hV:=[\hbeta\otimes\hI,\hI\otimes\halpha]$, the statement
about $\hB\otimes\hA$ in the theorem follows. The proof is
  complete.
\end{proof}

\subsection{Proof of Theorem~\ref{thm:mle}}

To prove Theorem~\ref{thm:mle}, we first list some properties of the function:
\begin{equation}
    \label{eq:pres_lik}
    h(\Omega,\hS)=-\log|\Omega|+\tr(\Omega \hS),
\end{equation}
where both $\Omega$ and $\hS$ are positive definite matrices. The first property is adapted from Theorem~7.6.6 of \cite{horn2012matrix}; and the second one can be proved by straightforward arguments, so we skip the proof.
\begin{proposition}
  \label{prop:pres_lik}
  Assume both $\Omega$ and $\hS$ are positive definite matrices of the same dimension.
  \begin{enumerate}
      \item [(i)] Fix $\hS$, the function $h(\Omega,\hS)$ is convex in $\Omega$ over the cone of positive definite matrices.
      \item [(ii)] Fix $\hS$, the second order Taylor expansion of $h(\Omega,\hS)$ around $\Omega$ is given by
      \begin{equation*}
          h(\bar\Omega,\hS) \approx h(\Omega,\hS) - \tr\left[(\Omega^{-1}-\hS)(\bar\Omega-\Omega)\right] +\tfrac{1}{2}[\vect(\bar\Omega-\Omega)]'(\Omega^{-1}\otimes\Omega^{-1})\vect(\bar\Omega-\Omega).
      \end{equation*}
  \end{enumerate}
\end{proposition}

\begin{proof}[Proof of Theorem~\ref{thm:mle}]
  To avoid the unidentifiability regarding $\Sigma_r$ and $\Sigma_c$, we make the convention that $\|\Sigma_r\|_F=1$. We will only prove that $\hat\hA_3=\hA+O_P(T^{-1/2})$, $\hat\hB_3=\hB+O_P(T^{-1/2})$, $\hat\Sigma_c=\Sigma_c+o_P(1)$, and $\hat\Sigma_r=\Sigma_r+o_P(1)$; and omit the rest of the proof, which is very similar with that of
  Theorem~\ref{thm:lse}.
  
  We first set up the notations for the proof. The matrices $\hA$, $\hB$, $\Phi:=\hB\otimes\hA$, $\Sigma_r$, $\Sigma_c$ and $\Sigma:=\Sigma_c\otimes\Sigma_r$ are used for the true parameters. We use $\hat\hA_3$, $\hat\hB_3$, $\hat\Phi_3=\hat\hB_3\otimes\hat\hA_3$, $\hat\Sigma_r$, $\hat\Sigma_c$, and $\hat\Sigma=\hat\Sigma_c\otimes\hat\Sigma_r$ to denote the MLE under \eqref{eq:loglik}. By the invariance property of MLE, finding the MLE of the covariance matrix $\Sigma=\Sigma_c\otimes\Sigma_r$ is equivalent as finding the MLE of the precision matrix $\Omega:=\Sigma^{-1}=\Omega_c\otimes\Omega_r$, where $\Omega_r=\Sigma_r^{-1}$ and $\Omega_c=\Sigma_c^{-1}$. Again $\hat\Omega$, $\hat\Omega_r$ and $\hat\Omega_c$ will denote the corresponding MLE under \eqref{eq:loglik}. Recall the definition of $\h{\mathcal{X}}$ and $\h{\mathcal{Y}}$ in \eqref{eq:design.mat}. For the unrestriced VAR(1) model \eqref{eq:var1} for $\vect(\hX_t)$, its log likelihood at the parameters $(\bar\Phi,\bar\Sigma)$, is 
  \begin{equation}
      \label{eq:var_loglik}
      \ell(\bar\Phi,\bar\Sigma)=-\tfrac{T-1}{2}\cdot h(\bar\Omega,\hS(\bar\Phi)),
  \end{equation}
  where $\bar\Omega:=\bar\Sigma^{-1}$, and $\hS(\bar\Phi):=(\h{\mathcal{Y}}-\bar\Phi\h{\mathcal{X}})(\h{\mathcal{Y}}-\bar\Phi\h{\mathcal{X}})'/(T-1)$. Let $\check\Phi:=\h{\mathcal{Y}}\h{\mathcal{X}}'(\h{\mathcal{X}}\h{\mathcal{X}}')^{-1}$, and $\check\hS:=\hS(\check\Phi)=(\h{\mathcal{Y}}-\check\Phi\h{\mathcal{X}})(\h{\mathcal{Y}}-\check\Phi\h{\mathcal{X}})'/(T-1)$. Note that $\hat\Phi$ and $\hat\hS$ are MLE for the unrestricted VAR(1) model \eqref{eq:var1}. 
  
  For a given $\bar\Omega$, the function $h(\bar\Omega,\bar\hS)$ is minimized at $\bar\hS=\check\hS$, with the minimum value $h(\bar\Omega,\check\hS)$.
  Let $\check\Omega:=\check\hS^{-1}$ be the MLE of $\Omega$ under \eqref{eq:var1}. We now prove that for any constant $c>0$
  \begin{equation}
      \label{eq:sigma_consistency}
      P\left[\inf_{\|\bar\Omega-\Omega\|_F\geq c}h(\bar\Omega,\check\hS)\leq h(\Omega,\hS(\Phi))\right]\rightarrow 0,
  \end{equation}
  which implies that $\hat\Sigma_c$ and $\hat\Sigma_r$ are consistent for $\Sigma_c$ and $\Sigma_r$.
  
  By Proposition~\ref{prop:pres_lik}, the multivariate Taylor expansion of the function $h(\bar\Omega,\check\hS)$ around $\check\Omega$ gives
  \begin{equation*}
      h(\bar\Omega,\check\hS)=h(\check\Omega,\check\hS)+\tfrac{1}{2}[\vect(\bar\Omega-\check\Omega)]'(\tilde\Omega^{-1}\otimes\tilde\Omega^{-1})\vect(\bar\Omega-\check\Omega),
  \end{equation*}
  where $\tilde\Omega=\check\Omega+\theta(\bar\Omega-\check\Omega)$ for some $0<\theta<1$. By the ergodic theorem, $\hS(\Phi)\stackrel{\hbox{\footnotesize a.s.}}{\longrightarrow}\Sigma$, and $\check\hS\stackrel{\hbox{\footnotesize a.s.}}{\longrightarrow}\Sigma$, and consequently $\check\Omega\stackrel{\hbox{\footnotesize a.s.}}{\longrightarrow}\Omega$. 
  It follows that for any $c>0$ that is small enough, both the following two events hold with probability approaching one:
  \begin{align*}
      & [\|\bar\Omega-\check\Omega\|\geq c/2 \hbox{ for all } \bar\Omega \hbox{ on the circle } \|\bar\Omega-\Omega\|_F=c],\\
      & [\lambda_{\min}(\tilde\Omega^{-1}\otimes\tilde\Omega^{-1})>\tfrac{1}{2}\lambda_{\min}^2(\Omega^{-1}) \hbox{ for all } \bar\Omega \hbox{ on the circle } \|\bar\Omega-\Omega\|_F=c],
  \end{align*}
  where $\lambda_{\min}(\cdot)$ denotes the minimum eigenvalue of a symmetric matrix. On the intersection of these two events, 
  \begin{equation*}
      \tfrac{1}{2}\vect(\bar\Omega-\check\Omega)(\tilde\Omega^{-1}\otimes\tilde\Omega^{-1})[\vect(\bar\Omega-\check\Omega)]'\geq \tfrac{1}{16}\lambda_{\min}^2(\Omega^{-1})\cdot c^2.
  \end{equation*}
  Since $h(\Omega,\hS(\Phi))\stackrel{\hbox{\footnotesize a.s.}}{\longrightarrow}h(\Omega,\Sigma)$ and $h(\check\Omega,\check\hS)\stackrel{\hbox{\footnotesize a.s.}}{\longrightarrow}h(\Omega,\Sigma)$, it follows that for any $c>0$ that is small enough
  \begin{equation*}
      P\left[\sup_{\|\bar\Omega-\Omega\|_F= c}h(\bar\Omega,\check\hS)\leq h(\Omega,\hS(\Phi))\right]\rightarrow 0.
  \end{equation*}
  Therefore, \eqref{eq:sigma_consistency} holds by the convexity of $h(\bar\Omega,\check\hS)$, when viewed as a function of $\bar\Omega$ (see Proposition~\ref{prop:pres_lik}).

 \medskip 
  We now prove that $\hat\hA_3=\hA+O_P(T^{-1/2})$, and $\hat\hB_3=\hB+O_P(T^{-1/2})$. Define the set $\mathcal{H}$ to be the collection of all positive definite matrices of the Kronecker product form: $$\mathcal{H}=\{\bar\Omega:\;\bar\Omega=
  \bar\Omega_c\otimes\bar\Omega_r \hbox{ for some $m\times m$ and $n\times n$ positive definite matrices } \bar\Omega_r \hbox{ and } \bar\Omega_c \}.$$ It suffices to show that for any seuqnce $c_T\rightarrow\infty$,
  \begin{equation}
      \label{eq:ABRootT}
      P\left[\inf_{\sqrt{T}\|\bar\Phi-\Phi\|_F\geq c_T} \inf_{\bar\Omega\in\mathcal{H}} h(\bar\Omega,\hS(\bar\Phi)) \leq \inf_{\bar\Omega\in\mathcal{H}} h(\bar\Omega,\hS(\Phi))\right] \rightarrow 0.
  \end{equation}
  Note that for any $\bar\Phi$, 
  \begin{equation*}
      (\h{\mathcal{Y}}-\bar\Phi\h{\mathcal{X}})(\h{\mathcal{Y}}-\bar\Phi\h{\mathcal{X}})' = (\h{\mathcal{Y}}-\check\Phi\h{\mathcal{X}})(\h{\mathcal{Y}}-\check\Phi\h{\mathcal{X}})' + (\bar\Phi-\hat\Phi)\h{\mathcal{X}}\h{\mathcal{X}}'(\bar\Phi-\check\Phi)'.
  \end{equation*}
  Let $\hat\Gamma_0:=\h{\mathcal{X}}\h{\mathcal{X}}'/(T-1)$ be the sample covariance matrix of $\vect(\hX_t)$, then the preceding equation can be written in the compact form
  \begin{equation*}
      \hS(\bar\Phi) = \check\hS + (\bar\Phi-\hat\Phi)\hat\Gamma_0(\bar\Phi-\check\Phi)',
  \end{equation*}
  which leads to
  \begin{align*}
      h[\bar\Omega,\hS(\bar\Phi)]&=h(\bar\Omega,\check\hS)+\tr[\bar\Omega(\bar\Phi-\hat\Phi)\hat\Gamma_0(\bar\Phi-\check\Phi)'] \\
      &\geq h(\bar\Omega,\check\hS) + \lambda_{\min}(\bar\Omega)\cdot\lambda_{\min}(\hat\Gamma_0)\cdot\|\bar\Phi-\check\Phi\|_F^2.
  \end{align*}
  According to \eqref{eq:sigma_consistency}, for any constant $c>0$, the following event holds with probability tending to 1:
  \begin{equation}
  \label{eq:4.2}
      \left[\inf_{\|\bar\Omega-\Omega\|_F\geq c}h(\bar\Omega,\check\hS)> h(\Omega,\hS(\Phi)) \geq \inf_{\bar\Omega\in\mathcal{H}} h(\bar\Omega,\hS(\Phi))\right].
  \end{equation}
  On the other hand, there exists a constant $C_1>0$, such that
  \begin{equation*}
      \inf_{\{\bar\Omega\in\mathcal{H}:\,\|\bar\Omega-\Omega\|_F\leq C_1\}} \lambda_{\min}(\bar\Omega) \geq \tfrac{1}{2}\lambda_{\min}(\Omega).
  \end{equation*}
  It follows that
  \begin{equation}
  \label{eq:4.3}
      \inf_{\{\bar\Omega\in\mathcal{H}:\,\|\bar\Omega-\Omega\|_F\leq C_1\}}h[\bar\Omega,\hS(\bar\Phi)]\geq      \inf_{\{\bar\Omega\in\mathcal{H}:\,\|\bar\Omega-\Omega\|_F\leq C_1\}}h[\bar\Omega,\check\hS]+ \tfrac{1}{2}\lambda_{\min}(\Omega)\cdot\lambda_{\min}(\hat\Gamma_0)\cdot\|\bar\Phi-\check\Phi\|_F^2.
  \end{equation}
  Since $\hat\Gamma_0\stackrel{\hbox{\footnotesize a.s.}}{\longrightarrow}\Gamma_0$ and $\check\Phi=\Phi+O_P(1/\sqrt{T})$, we know
  \begin{equation}
  \label{eq:4.4}
      P\left[\inf_{\sqrt{T}\|\bar\Phi-\Phi\|_F\geq c_T}\lambda_{\min}(\hat\Gamma_0)\cdot\|\bar\Phi-\check\Phi\|_F^2 \geq \tfrac{1}{2}\lambda_{\min}(\Gamma_0)\cdot c_T^2/T\right]\rightarrow 1.
  \end{equation}
  Consider the function $h(\bar\Omega,\hS(\Phi))$. Since
  \begin{equation*}
      \hS(\Phi) = \check\hS + (\Phi-\hat\Phi)\hat\Gamma_0(\Phi-\check\Phi)',
  \end{equation*}
  it holds that
  \begin{equation}
   \label{eq:4.5}   
  \begin{aligned}
      \inf_{\{\bar\Omega\in\mathcal{H}:\,\|\bar\Omega-\Omega\|_F\leq C_1\}}h[\bar\Omega,\hS(\Phi)]&=\inf_{\{\bar\Omega\in\mathcal{H}:\,\|\bar\Omega-\Omega\|_F\leq C_1\}} \left\{h(\bar\Omega,\check\hS)+\tr[\bar\Omega(\Phi-\check\Phi)\hat\Gamma_0(\Phi-\check\Phi)']\right\} \\
      & = \inf_{\{\bar\Omega\in\mathcal{H}:\,\|\bar\Omega-\Omega\|_F\leq C_1\}}h(\bar\Omega,\check\hS) + O_P(T^{-1}).
  \end{aligned}
  \end{equation}
  Combining \eqref{eq:4.3} \eqref{eq:4.4} and \eqref{eq:4.5}, and noting that $c_T\rightarrow\infty$, we have established that with probability converging to 1,
  \begin{equation*}
  \label{eq:4.6}
      \inf_{\sqrt{T}\|\bar\Phi-\Phi\|_F\geq c_T}\inf_{\{\bar\Omega\in\mathcal{H}:\,\|\bar\Omega-\Omega\|_F\leq C_1\}}h[\bar\Omega,\hS(\bar\Phi)] >     \inf_{\{\bar\Omega\in\mathcal{H}:\,\|\bar\Omega-\Omega\|_F\leq C_1\}}h[\bar\Omega,\hS(\Phi)].
  \end{equation*}
  The preceding equation, together with \eqref{eq:4.2}, implies \eqref{eq:ABRootT}; and the proof is therefore complete.
\end{proof}

\subsection{Proof of Corollaries}

\begin{proof}[Proof of Corollary~\ref{thm:eff}]
  From Theorem~\ref{thm:lse} and Theorem~\ref{thm:mle}, we have 
  \begin{align*}
    \Xi_2 & = [\E(\hW_t\hW_t')+\hgamma\hgamma']^{-1}\E(\hW_t\Sigma\hW_t')
            [\E(\hW_t\hW_t')+\hgamma\hgamma']^{-1},\\
    \Xi_3 & = [\E(\hW_t\Sigma^{-1}\hW_t')+\hgamma\hgamma']^{-1}\E(\hW_t\Sigma^{-1}\hW_t')
            [\E(\hW_t\Sigma^{-1}\hW_t')+\hgamma\hgamma']^{-1}.
  \end{align*}
  Let $\Sigma=\hQ\diag\{\lambda_1,\lambda_2,\ldots,\lambda_{mn}\}\hQ'$
  be the spectral decomposition of $\Sigma$, and use
  $\hu_1,\hu_2,\ldots,\hu_{mn}$ to denote the $mn$ columns of the
  matrix $\hW_t\hQ$. We further let $\Theta_i=\E(\hu_i\hu_i')$, and
  note that $\Theta_i$ is a symmetric positive semi-definite
  matrix. Then the preceding equation becomes
  \begin{equation*}
    \begin{aligned}
      \Xi_2 & = \left(\sum_{i=1}^{mn}\Theta_i+\hgamma\hgamma'\right)^{-1}
      \left(\sum_{i=1}^{mn}\lambda_i\Theta_i\right)
      \left(\sum_{i=1}^{mn}\Theta_i+\hgamma\hgamma'\right)^{-1},\\
      \Xi_3 & = \left(\sum_{i=1}^{mn}\lambda_i^{-1}\Theta_i+\hgamma\hgamma'\right)^{-1}
      \left(\sum_{i=1}^{mn}\lambda_i^{-1}\Theta_i\right)
      \left(\sum_{i=1}^{mn}\lambda_i^{-1}\Theta_i+\hgamma\hgamma'\right)^{-1}.
    \end{aligned}
  \end{equation*}
  Since $\hW_t'\hgamma=\hzero$, it holds that $\Theta_i\hgamma=\hzero$
  for all $1\leq i\leq mn$. It follows that
  \begin{equation}
    \label{eq:eff1}
    \begin{aligned}
      \Xi_2+\hgamma\hgamma' & = \left(\sum_{i=1}^{mn}\Theta_i+\hgamma\hgamma'\right)^{-1}
      \left(\sum_{i=1}^{mn}\lambda_i\Theta_i+\hgamma\hgamma'\right)
      \left(\sum_{i=1}^{mn}\Theta_i+\hgamma\hgamma'\right)^{-1}\\
      \Xi_3 + \hgamma\hgamma' & = \left(\sum_{i=1}^{mn}\lambda_i^{-1}\Theta_i+\hgamma\hgamma'\right)^{-1}.
    \end{aligned}
  \end{equation}
  To simplify the long equations, we let
  $\Theta_{mn+1}=\hgamma\hgamma'$, $\lambda_{mn+1}=1$, and make the
  convention that all the sums over $i$ runs from $i=1$ to
  $i=mn+1$. The equation \eqref{eq:eff1} becomes
  \begin{equation}
    \label{eq:eff2}
    \begin{aligned}
      \Xi_2+\hgamma\hgamma' & = \left(\sum_i\Theta_i\right)^{-1}
      \left(\sum_i\lambda_i\Theta_i\right)
      \left(\sum_i\Theta_i\right)^{-1},\\
      \Xi_3 + \hgamma\hgamma' & = \left(\sum_i\lambda_i^{-1}\Theta_i\right)^{-1}.
    \end{aligned}
  \end{equation}
  From \eqref{eq:eff2}, we see that in order to show that
  $\Xi_2\geq\Xi_3$, it suffices to show that
  \begin{equation}
    \label{eq:eff3}
    \sum_i\lambda_i^{-1}\Theta_i - \left(\sum_i\Theta_i\right)
    \left(\sum_i\lambda_i\Theta_i\right)^{-1}
    \left(\sum_i\Theta_i\right)
  \end{equation}
  is positive semi-definite. For this purpose, we construct the matrix
  \begin{equation}
    \label{eq:eff4}
    \begin{pmatrix}
      \sum_i\lambda_i^{-1}\Theta_i & \sum_i\Theta_i \\
      \sum_i\Theta_i & \sum_i\lambda_i\Theta_i
    \end{pmatrix} = \sum_i
    \begin{pmatrix}
      \lambda_i^{-1}\Theta_i & \Theta_i \\
      \Theta_i & \lambda_i\Theta_i
    \end{pmatrix}.
  \end{equation}
  Since $\Theta_i$ is positive semi-definite, each term on the right
  hand side of \eqref{eq:eff4} is also positive semi-definite, and so
  is the sum in \eqref{eq:eff4}. If we view the matrix in
  \eqref{eq:eff4} as a covariance matrix, then the matrix in
  \eqref{eq:eff3} is the conditional covariance matrix of the first
  half given the second half, so it is positive semi-definite, and the
  proof is complete.
\end{proof}

\begin{proof}[Proof of Corollary~\ref{thm:test}]
Following standard theory of
multivariate ARMA models \citep{hannan:1970,dunsmuir:1976}, the
conditions of Theorem~\ref{thm:clt1} guarantees that $\hat\Phi$
converges to a multivariate normal distribution:
\nocite{deistler:1978}
\begin{equation*}
  \sqrt{T}\,\mathrm{vec}(\hat\Phi-\hB\otimes\hA) \Rightarrow
  N(\hzero,\Gamma_0^{-1}\otimes\Sigma),
\end{equation*}
where $\Sigma$ is the covariance matrix of $\vect(\hE_t)$, and
$\Gamma_0$ is given in \eqref{eq:autocov}. 
Recall that $\halpha=\vect(\hA)$ is a unit vector, $\hubeta=\vect(\hB)$, and $\hbeta_1$ is the normalized version of $\hubeta$. Since $\tilde\Phi$ is the rearranged version of $\hat\Phi$, it follows that
\begin{equation*}
    \sqrt{T}\,\mathrm{vec}(\tilde\Phi-\halpha\hubeta') \Rightarrow
  N(\hzero,\Xi_1).
\end{equation*}
According to \eqref{eq:4}, it holds that
  \begin{equation*}
    \hat\hD=\tilde\Phi-\hat\halpha\hat\hubeta'  = (\hI-\halpha\halpha')(\tilde\Phi-\halpha\hubeta')(\hI-\hbeta_1\hbeta_1') + o_P(T^{-1/2}).
  \end{equation*}
Recall that $\hP$ is defined as $(\hI-\hbeta_1\hbeta_1')\otimes(\hI-\halpha\halpha')$, so after taking vectorization on both sides, 
\begin{equation*}
    \vect(\hat \hD) = \hP\,\vect(\tilde\Phi-\halpha\hubeta') + o_P(T^{-1/2}),
\end{equation*}
and consequently
\begin{equation*}
    \sqrt{T}\,\vect(\hat \hD)\Rightarrow N(\hzero,\hP\,\Xi_1\hP).
\end{equation*}
Note that the matrix $\hP_1$ is an orthogonal projection matrix with rank $(m^2-1)(n^2-1)$, therefore it follows that
\begin{equation*}
      T\cdot\vect(\hat\hD)'\;(\hP\,\Xi_1\,\hP)^+\;\vect(\hat\hD)\Rightarrow\chi^2_{(m^2-1)(n^2-1)},
\end{equation*}
and the proof is complete.  
\end{proof}


\end{document}